\documentclass[12pt,epsf]{article}  
\usepackage{graphicx,amsmath}

\usepackage{amsmath,amsfonts,latexsym,amssymb}
\usepackage{verbatim,amsthm,curves,graphics}
\usepackage{mathrsfs}
\usepackage[]{inputenc}
\usepackage[T1]{fontenc}
\usepackage{hyperref}
\usepackage{color}
\usepackage{lmodern}
\usepackage[toc,page]{appendix}
\usepackage{fullpage}
\usepackage{authblk}
\usepackage[vcentermath]{youngtab}
\usepackage{ytableau}

\usepackage{enumitem}
\usepackage{amsthm}
\usepackage{amsfonts}
\usepackage{amssymb}
\usepackage{amsthm}
\usepackage{amsmath}    
\usepackage{hyperref}
\usepackage{bbm}
\usepackage{graphicx}
\usepackage{latexsym}
\usepackage{mathrsfs}
\usepackage{bbm}
\usepackage{slashed}
\usepackage{verbatim}
\usepackage{color}
\usepackage{moodle}
\usepackage{tikz}
\usepackage[normalem]{ulem} 
\usepackage{multicol}
\usepackage{geometry}
\DeclareMathAlphabet\EuFrak{U}{euf}{m}{n}	
\SetMathAlphabet\EuFrak{bold}{U}{euf}{b}{n}	

\newcommand{\cA}{\mathcal{A}}
\newcommand{\cB}{\mathcal{B}}
\newcommand{\cC}{\mathcal{C}}

\newcommand{\cH}{\mathcal{H}}

\newcommand{\cM}{\mathcal{M}}
\newcommand{\cN}{\mathcal{N}}

\newcommand{\cR}{\mathcal{R}}

\newcommand{\cV}{\mathcal{V}}

\newcommand{\sH}{\mathscr{H}}

\newcommand{\sK}{\mathscr{K}}

  \theoremstyle{plain}
  \newtheorem{definition}{Definition}[section]
  \newtheorem{theorem}[definition]{Theorem}
  \newtheorem{proposition}[definition]{Proposition}
  \newtheorem{corollary}[definition]{Corollary}
  \newtheorem{lemma}[definition]{Lemma}
  
  \theoremstyle{definition}
  \newtheorem{remark}[definition]{Remark}
  \newtheorem{example}[definition]{Example}

\begin{document}

\author{Stefan Hollands}
\affil{\small Institute for Theoretical Physics, Leipzig University, Br\" uderstrasse 16, 04103 Leipzig, and MPI-MiS, Inselstrasse 22, 04103, Leipzig, Germany, stefan.hollands@uni-leipzig.de}
\author{Alessio Ranallo}
\affil{\small Dipartimento di Matematica,
Universit\`a di Roma Tor Vergata,
Via della Ricerca Scientifica, 1, I-00133 Roma, Italy, ranallo@mat.uniroma2.it}

\title{Complexity in algebraic QFT}

\maketitle

\begin{abstract}
We consider a notion of complexity of quantum channels in relativistic continuum quantum field theory (QFT) defined by the 
distance to the trivial (identity) channel. Our distance measure is based on a specific divergence between 
quantum channels derived from the Belavkin-Staszewski (BS) divergence. We prove in the prerequisite generality 
necessary for the algebras in QFT that the corresponding complexity has several reasonable properties: (i) the complexity of 
a composite channel is not larger than the sum of its parts, (ii) it is additive for channels localized in spacelike separated regions, (iii) it is convex, (iv)
for an $N$-ary measurement channel it is $\log N$, (v) for a conditional expectation associated with an inclusion of QFTs 
with finite Jones index it is given by $\log (\text{Jones Index})$. The main technical tool in our work is a new variational principle 
for the BS divergence.
\end{abstract}

\section{Introduction}

A  problem of both theoretical and practical interest in quantum information theory is to assess the ``complexity'' of a quantum state- or operation. A natural approach is to 
take as a measure of complexity the minimum number of operations from an underlying set considered as ``basic'' \cite{bernstein1993quantum, nielsen2005geometric,nielsen2006quantum}. 
Typical results in this context include 
bounds on the growth of complexity under time evolution, see e.g. \cite{haferkamp2022linear,Li:2022hkh}. 
There are also proposals in the context of the AdS-CFT correspondence, linking the growth of complexity 
of a state in the boundary quantum field theory (QFT) to various geometric quantities in the bulk, see e.g. \cite{brown2018second,stanford2014complexity,brown2015complexity}. 

One may ask how to define a notion of complexity directly in a relativistic continuum QFT without reference to holographic ideas. 
In QFT, one faces the immediate problem to identify a suitable set of basic operations with respect to which the complexity of a composite operation is 
supposed to be assessed. If one wants to maintain a close analogy to ideas such as \cite{haferkamp2022linear,Li:2022hkh}, it appears that one would have to specify a preferred set of local quantum field operators within some lattice regularization of the theory. For instance, for Gaussian field theories, concrete proposals include \cite{chapman2018toward,jefferson2017circuit}, which lead to divergent results\footnote{Corresponding divergences also naturally arise in holographic approaches.} as the UV-cutoff is removed. 
For a general QFT it seems to us that both the operator basis and lattice regularization 
would be highly non-unique in view of the universality phenomenon, and at any rate also clearly breaking the relativistic invariance of the theory. 

One approach to this issue is to take a broader view of the problem, departing from the notion of basic operation and focussing attention instead on a suitable notion of ``distance'', $D(S\|T)$, between two channels $S,T$. Complexity would then be defined as $c(T)=D(id\|T)$, the distance to the trivial (identity) channel. Of course,  
one would like specific properties from $D$ to connect to the idea of complexity. 
Natural requirements would be: 
\begin{itemize}
\item
Subadditivity: $c(T_1 \circ T_2) \le c(T_1) + c(T_2)$ expressing that the complexity of a composite channel is 
not bigger than the sum of its parts. In particluar, for a 1-parameter (Markov-) semi-group $T_t, t \ge 0$ of channels, 
we automatically get at most linear growth in time, $c(T_{Nt_0}) \le C_0 N$. 
\item 
Locality: $c(T_1 \circ T_2) = c(T_1) + c(T_2)$ if  $T_1$ and $T_2$ are localized in spacelike related parts of the system. This 
expresses that the complexity $c$ respects Einstein causality/locality. 
\item
Convexity: Thinking about performing operations $T_i$ randomly with some probabilities $p_i$ it is natural to ask that 
$c(\sum p_i T_i) \le \sum p_i c(T_i)$. 
\end{itemize}
One way to obtain a notion of channel divergence, hence $c$, is to start with a corresponding divergence $D(\varphi \| \psi)$ in the ordinary sense between quantum 
states\footnote{
\label{foot1}
In this paper we follow the conventions in operator algebras that a state $\psi$ is a positive functional on the observable algebra. 
For matrix algebras, $\psi(m) = {\rm Tr}(m\rho_\psi)$, where $\rho_\psi$ is the corresponding density matrix, see section \ref{sec:conventions} for our conventions.} $\varphi, \psi$ (see e.g. \cite{Wil}), by considering how much the actions of $T$ and $S$ on a state can deviate as quantified by this divergence.
With this idea in mind, a naive guess would be to consider $\sup_{\psi} D(\psi \circ S \| \psi \circ T)$, where the maximization 
is over all normalized states of the system and $\psi \circ T$ is the action of the channel on the state $\psi$, viewed in this paper as an expectation functional on observables, 
see footnote \ref{foot1}.
However, as is well-known, this notion is actually inadequate for quantum systems because 
one can obtain refined information about the action of channels by coupling the systems in question to an ancillary system and considering 
states that have a suitably engineered entanglement between the original system and the ancillary system. So one should define instead
\begin{equation}
D(S\|T) = \sup_{\psi, \cA} D(\psi \circ (S \otimes id_\cA) \| \psi \circ (T \otimes id_\cA))
\end{equation}
where the maximization is now over all states of the original observable algebra $\cM$ tensored with the ancillary algebra $\cA$. This is the definition that we shall also adopt, up to 
some technical caveats related to the fact that we will be dealing with  von Neumann algebras of a sufficiently general type as appropriate for QFT.\footnote{In the setting of von Neumann algebras, it is most natural to model the ancillary system by another von Neumann algebra, $\cA$, and it does not seem 
natural to restrict the nature of that system, e.g. by imposing that $\cA$ should have a particular type such as I$_n$. Then we must also have an enlarged Hilbert space on which both the original von Neumann algebras as well as the ancillary algebra $\cA$ acts, i.e. we must consider bi-modules of von Neumann algebras, see e.g. \cite{sauvageot1983produit,Lon18} and references therein.}
 
Of course the main question is what $D$ we should start from. One possibility might be a geometric approach along the lines of 
\cite{nielsen2005geometric,nielsen2006quantum}. In \cite{Li:2022exc}, on the other hand, the authors propose to use a particular 
quantum version \cite{de2021quantum} of the classical ``Wasserstein-distance'', see e.g. \cite{otto2000generalization}, and derive several convincing properties of the corresponding notion of complexity, including the ones listed above. The quantum Wasserstein distance as defined by \cite{de2021quantum} is 
for finite dimensional systems with Hilbert space of the form $(\mathbb{C}^d)^{\otimes N}$. While it may be possible to generalize 
it to von Neumann algebras of type III appearing in QFT \cite{Buc87}, we proceed differently here
and work with the so-called Belavkin-Staszewski (BS) divergence \cite{belavkin1982c} $D_{BS}$. 
That divergence has been considered\footnote{Divergences of this type have also been considered 
in the context of QFT in \cite{furuya2021monotonic}.} recently in the context of channel discrimination by \cite{Wil20} and is 
\begin{equation}
D_{BS}(\varphi\| \psi) = {\rm Tr}(\rho_\varphi^{} \log[\rho_\varphi^{1/2} \rho_\psi^{-1} \rho_\varphi^{1/2}])
\end{equation}
for matrix algebras.
Here, $\rho_\psi$ is the density matrix representing the expectation functional $\psi$, i.e. $\psi(m) = {\rm Tr}(m\rho_\psi)$ for all $m\in \cM$. 
A generalization to arbitrary von Neumann algebras is possible \cite{Hia21,Hia19}.
Our reason for considering $D_{BS}$ is that \cite{Fan21} (see also \cite{fang2020chain}) have shown 
that the corresponding channel divergence has the desired subadditivity property\footnote{In fact, these authors 
study the corresponding Renyi-type ``geometric divergences'', of which the BS-divergence arises as a limit.} in the finite dimensional setting, 
contrary to some other well-known divergences such as, say, 
the more commonly used Araki-Umegaki relative entropy \cite{Ara76}. In this work, we analyze the corresponding channel divergence in 
the context of general von Neumann algebras, and prove that the corresponding notion of complexity has the above properties. 

Hence, it is of potential use in QFT. We also prove a number of further properties:
\begin{enumerate}
\item If $T(a) = uau^*$ is the channel corresponding to a non-trivial local unitary, then $c(T)=\infty$.
\item If $T(a) = \rho(a)$ is the channel corresponding to a non-trivial representation of the QFT (``charge superselection sector''), then $c(\rho) = \infty$. 
\item If $M(a) = \sum_{i=1}^N e_i a e_i$ is the channel corresponding to a local $N$-ary von Neumann measurement, then $c(M) = \log N$.
\item Let $E_\rho$ be the (minimal) conditional expectation from $\cA(O)$ to $\rho(\cA(O))$ where $\rho$ is a charge superselection sector 
(charged representation), then  
\begin{equation}
c(E_\rho) = \log(\text{Jones Index}) = \log d_\rho^2,
\end{equation}
where we mean 
the Jones index \cite{Jon83} of the inclusion $\rho(\cA(O)) \subset \cA(O)$, and where $d_\rho$ is the statistical (or ``quantum-'') dimension of the sector.\footnote{The last 
equality is a direct consequence of the index-statistics theorem \cite{Lon89}.}
\end{enumerate}
Items 1), 2) are basically negative results, but perhaps not totally unreasonable if we remember that any local operation in a continuum QFT (i.e. an operation in 
a finite spacetime region) must still involve an infinite number of degrees of freedom. The channels in items 3), 4) are conditional expectations. 
This suggest that these are to be regarded as the basic operations in QFT. 

Particular measurements in 3) implementing the idea of ``setting individual q-bits'' can be constructed trivially as follows. Imagine the QFT has a ``basic'' real scalar field $\phi$ and 
consider a cube of side-length $\delta$ in a time slice. Let $f$ be a non-negative testfunction supported in the cube and let $S= \int \phi(0,{\bf x}) f({\bf x}) d^{n-1}{\bf x}$, where $n$
is the dimension of spacetime, and let $p_\pm$ be the projectors corresponding to a positive/negative measurement of $S$. Shifting the cube periodically in the $n-1$ spatial directions we can obtain a finite lattice $\Lambda$
with corresponding projections $p_{{\bf x},\sigma}, \sigma = \pm, {\bf x} \in \Lambda$ associated with each point $\bf x$ of the (dual) lattice. Then 
we can define projections $e(\{ \sigma \}) = \prod_{{\bf x} \in \Lambda} p_{{\bf x},\sigma({\bf x})}$, each corresponding to measuring 
a particular lattice configuration $\{ \sigma \}$, e.g. 

\pagebreak
\begin{figure}[h]
\begin{center}
 {\small
\begin{ytableau}
+&-&-&+&-\\
-&-&+&+&+\\
-&+&-&-&-\\
-&+&+&-&+\\
-&+&+&-&+
 \end{ytableau}
 }
 \end{center}
\end{figure}
\noindent
The complexity of the corresponding measurement channel is clearly 
\begin{equation}
c(M) = \frac{{\rm vol}(\Lambda)}{\delta^{n-1}} \log 2
\end{equation}

As an example of item 4), consider the QFT of a real $N$-component free complex Klein-Gordon quantum field $\phi_I(x), I=1, \dots, N$. 
We consider as observables the charge neutral operators (gauge invariant observables) under the $SU(N)$-symmetry. Consider a state $\Psi$
in the Hilbert space which is in some non-trivial representation $R$ of $SU(N)$. Then $\Psi$ cannot be generated from the vacuum $\Omega$
by the action of any charge neutral operator $a$, so the representation of charge neutral operators built on $\Psi$ is not unitarily equivalent 
to the vacuum representation. In fact, by DHR theory \cite{doplicher1971local,doplicher1969fields}, there exists an endomorphism 
$\rho$ of the algebra of charge neutral operators such that
\begin{equation}
\langle \Omega, \rho(a) \Omega \rangle = \langle \Psi, a \Psi \rangle\quad \text{for all charge neutral $a$,}
\end{equation}
and $\rho$ implements the charged sector with representation $R$. The 
statistical dimension $d_\rho$ of this $\rho$ equals the dimension $d_R$ of the representation $R$ in this case, e.g. $d_\rho = N^2-1$ if 
$R$ is the adjoint representation. Details of this construction are given in example \ref{DHRexp} below. For low dimensional QFTs, $d_\rho$ does not have to be integer.

In fact, the Jones index in 4) ($=d_\rho^2$) 
is restricted to the set $\{ 4\cos^2(\pi/n) : n=3,4,5,\dots\} \cup [4,\infty]$ by Jones' theorem \cite{Jon83}, the smallest non-trivial value of which is 
$2$, realized e.g. by the sector $\rho$ of the $(4,3)$ minimal model (Ising) with quantum dimension $d_\rho=\sqrt{2}$. We conjecture that for any localized channel $T$\footnote{This would imply in particular that for any 1-parameter Markov semi-group $T_t$ of channels $c(T_t)$ is necessarily discontinuous at 
$t=0$.} 
\begin{equation}
\text{Either} \quad c(T) \ge \log 2 \quad \text{or} \quad T=id \quad (\text{conjecture}),
\end{equation}
which is reminiscent of the famous Landauer bound. 

This paper is organized as follows. In section 2, we first recall the theory of $f$-divergences and operator means for states on 
von Neumann algebras and introduce our main technical tool, 
a variational characterization of $D_{BS}$ (proposition \ref{BS}). In section 3 we introduce 
$D_{BS}$ for channels of von Neumann algebras of general type, and prove some basic properties. In section 4, we apply these results to QFT. 

\section{Preliminaries}

\subsection{Von Neumann algebra terminology and basic objects}
\label{sec:conventions}
See e.g. \cite{Dix11} as a general general reference. 
\begin{itemize}
\item 
Von Neumann algebra: 
A von Neumann algebra $\cM$ is a closed $\ast-$subalgebra of the algebra of bounded operators $B(\sH)$ on a 
Hilbert space $\sH$ in the weak operator topology. The weak topology is defined by the matrix elements, i.e. the open neighborhoods are
$N(\xi_i,\eta_i,\varepsilon,a) = \{b : \  |\langle\xi_i, (b-a) \eta_i \rangle|<\varepsilon, i=1, \dots, n\}$, where $a \in B(\sH), \xi_i,\eta_i \in \sH, \varepsilon>0$.  
All Hilbert spaces appearing in this paper are assumed to be separable. The squared norm $\|m\|^2$ of an operator $m \in \cM$ 
is defined to be the supremum of the spectrum $\sigma(mm^*)$ of the positive operator $mm^*$. The subset of all such operators is 
denoted by $\cM_+$ (positive part).

\item 
Any finite-dimensional von Neumann algebra is isomorphic to $\oplus_{i=1}^N M_{n_i}(\mathbb{C})$ for some $n_i$, where $M_n(\mathbb{C})$ is the algebra of 
complex $n \times n$ matrices. 

\item
(Bi-)Commutant:
An equivalent characterization of von Neumann algebra is $\cM''=\cM$, where $\cM':=\left\{x \in B(\sH) \ | \ xm=mx \ \forall m \in \cM \right\}$ is the commutant 
of $\cM$ in $B(\sH)$, and $\cM''=(\cM')'$ is the bicommutant. A von Neumann algebra is called a factor if $\cM \cap \cM' = \mathbb{C}1$. One denotes by 
$\cA \vee \cB = (\cA \cup \cB)''$ the von Neumann algebra generated by sets of bounded operators $\cA,\cB$.

\item
States:
A state is a linear, positive, normal, normalized functional $\psi:\cM \to \mathbb{C}$, where positive means $\psi(mm^*) \ge 0$ for all $m \in \cM$ and normalized means $\psi(1)=1$. 
A linear functional $\psi$ is called normal if it is ultra-weakly continuous, and a positive linear functional $\psi$ is called faithful if $\psi(mm^*)=0 \Longrightarrow m=0$. 
The set of normal states is also denoted by $\cM_{*,+}$. The existence of 
a normal faithful positive linear functional is guaranteed since we are assuming that $\sH$ is separable. On a matrix algebra every state is of the form 
\begin{equation}
\label{densitym}
\psi(m) = {\rm Tr}(m\rho_\psi)
\end{equation}
 for a unique density matrix $\rho_\psi$.

\item 
Channels:
Generalizing the notion of state, a 
channel $T:\cM \to \cN$ is a normal, positive, unital (meaning $T(1)=1$) linear map which is also completely positive, meaning that 
$T \otimes id: \cM \odot B(\sK) \to \cN \odot B(\sK)$ is positive, where $\odot$ means the algebraic tensor product and where $\sK$ is any Hilbert space. 
If $\psi$ is a state on $\cN$, then $\psi \circ T(m):=\psi(T(m))$ is a state on $\cM$, and $\psi \mapsto \psi \circ T$ corresponds to the dual action of channels on 
states (Schr\" odinger picture). In much of the quantum information theory literature, the Schr\" odinger picture is considered, but of course this is just a matter of convention. 
For finite dimensional von Neumann algebras $\cN,\cM$, the action of $T$ in the Schrödinger picture may also be thought of as an action $T^+$ on density matrices
\eqref{densitym}, 
\begin{equation}
{\rm Tr}(T^+(\rho_\psi) m) = {\rm Tr}(\rho_\psi T(m)) \equiv \psi \circ T(m) \quad m \in \cM.
\end{equation}
Then $T^+$ is completely positive and trace preserving (corresponding to $T(1)=1$).

\item 
Standard form:
A vector $\Omega \in \sH$ is called cyclic if $\cM \Omega$ is dense in $\sH$ in the strong topology, and it is called separating if $m\Omega = 0 \Longrightarrow m=0$. 
Such a representation on $\sH$ of $\cM$ and vector can always be obtained by the GNS-representation of a faithful normal state $\omega$. A cyclic and separating vector 
is also called standard and a representation of $\cM$ on a Hilbert space with standard vector is called a standard representation. Associated with $\Omega$ is an anti-linear 
involution $J$ on $\sH$ such that $J\Omega = \Omega$ and $J\cM J = \cM'$ called the modular conjugation. The closure of the set of vectors of the form $aJaJ\Omega, a \in \cM$
is called ``natural cone'' and is also denoted as $L^2(\cM,\Omega)_+ \subset \sH$. 

\item
Conditional expectations:
If $\cN \subset \cM$ is a von Neumann subalgebra, then a conditional expectation $E:\cM \to \cN$ is a channel such that $E(n_1mn_2) = n_1 E(m) n_2$ for all 
$m \in \cM, n_i \in \cN$. The index $\lambda_E \in [1,\infty]$ of a conditional expectation is the infimum  over all positive real numbers $\lambda$ such 
that the $E(mm^*) \ge \lambda^{-1} mm^*$ for all $m \in \cM$. 

\item
Jones-index:
Assume that $\cN, \cM$ are factors such that there exists a conditional expectation $E:\cM \to \cN$. If $\lambda_E<\infty$, there exists a unique $E_0$ 
called ``minimal conditional expectation'' \cite{Hia88} such that $\lambda_{E_0}$ is minimal, and in 
such a case $\lambda_{E_0} =: [\cM:\cN]$ is called the Jones-Kosaki index \cite{Kos86a,Jon83} of the inclusion. Otherwise we set $[\cM:\cN]=\infty$. 

\item
$L^p$-space:
One can construct so-called ``non-commutative $L^p$ spaces'' ($p \in [1,\infty]$) interpolating between the space of normal functionals on $\cM$ and $\cM$ itself. 
They are defined relative 
to some standard vector $\Omega$ and denoted as $L^p(\cM,\Omega)$, see \cite{araki1982positive}. One has $L^2(\cM,\Omega) = \sH$. Beyond this, we will only need $L^\infty(\cM,\Omega)$ which is a linear subspace of $\sH$, and we shall 
mainly used the following characterization of this space \cite{araki1982positive}: As a vector space $L^\infty(\cM,\Omega) = \cM \Omega$. The Banach space norm is $\| \xi \|_{L^\infty(\cM,\Omega)} 
= \|m\|$ where $m \in \cM$ is the unique element such that $\xi = m\Omega$.

\item Opposite algebra: The opposite algebra $\cM^{op}$ of a von Neumann algebra $\cM$ is identical as a vector space with $*$-operation, but has the reversed 
product $m_1^{op}m_2^{op} = (m_2m_1)^{op}$.
\end{itemize}

\subsection{Maximal $f$-divergence for bounded operators}

See \cite{Bha97,Hia10, Hia21, hiai2022pusz}
as general references. Central to the concept of operator mean and the divergences studied in this paper are the notions of operator monotone- 
and operator convex functions. 

\begin{definition}
Let $I\subset \mathbb{R}$ be an interval. $f:I\rightarrow \mathbb{R}$ is said to be 
\begin{itemize}
    \item operator monotone if $f(A)\leq f(B)$ whenever $A,B \in B(\sH)$ are self adjoint operators on a Hilbert space such that $A \le B$ and that their spectra satisfy 
    $\sigma(A),\sigma(B) \subset I$;
    \item operator convex if $f(\lambda A + (1-\lambda)B)\leq \lambda f(A) + (1-\lambda)f(B), \ \forall \lambda \in (0,1)$ whenever $A,B \in B(\sH)$, with $\sigma(A),\sigma(B) \subset I$.
    \end{itemize}
\end{definition}

\begin{remark}
Let $t_0 \in (0,\infty]$ and $f:[0,t_0)\rightarrow \mathbb{R}$. Then ($f$ is operator convex and $f(0)\leq 0$) if and only if ($\frac{f(t)}{t}$ is operator monotone on $(0,t_0)$). 
Furthermore, if $f:[0,t_0)\rightarrow \mathbb{R}$ is operator monotone, then it is also operator concave. 
While the converse is not true, it is the case that ($f$ operator concave and $f(t)\geq 0$ for all $t \in [0,\infty)$) implies ($f$ is operator monotone on $[0,\infty)$).
\end{remark}

\begin{example}
On $[0,\infty)$, the function $t^\alpha$ is operator monotone if and only $\alpha \in [0,1]$. $t^\alpha$ is operator convex if and only if $\alpha \in [-1,0]\cup [1,2]$. 
The function $\log(t)$ is operator monotone on $(0,\infty)$.
\end{example}

The following well-known representation \eqref{frep} allows one to reduce many constructions involving operator monotone functions to a certain weighted averages of a special 
operator monotone function. 
Consider a continuous operator monotone function $f$ on $[0,\infty)$, let $a=f(0), \ b=f'(\infty):=\lim_{t\rightarrow \infty}\frac{f(t)}{t}$. There exists a unique finite positive Radon measure 
$\mu$ on $[0,\infty)$, such that 
\begin{equation}
\label{frep}
f(t)=a+bt+\int_{(0,\infty)} \frac{(1+s) t}{t+s}d\mu(s).
\end{equation}

\begin{definition}[Kubo--Ando means, \cite{KA80}]
Consider a binary operation $\sigma$ on $B(\sH)_+$ (non-negative self-adjoint bounded operators), i.e. $\sigma: B(\sH)_+ \times B(\sH)_+ \rightarrow B(\sH)_+$. 
We write $\sigma(A\times B)=: A\sigma B \in B(\sH)_+$. $\sigma$ is called a Kubo--Ando connection if, for all $A,B,C,D \in B(\sH)_+$, the following hold 
\begin{enumerate}
    \item joint monotonicity, i.e. $A\leq C, \ B\leq D,\ \text{then} \ A\sigma B\leq C\sigma D;$
    \item transformer inequality, i.e. $C(A\sigma B)C\leq (CAC)\sigma (CBC);$
    \item upper semicontinuity, i.e. whenever $A_n \downarrow A$, $B_n \downarrow B$ strongly, then $A_n\sigma B_n \downarrow A\sigma B,$ strongly.
\end{enumerate}

Moreover $\sigma$ is called a (Kubo--Ando operator) mean if the above hold and 

\begin{enumerate}[resume]
    \item normalization, i.e. $$I_\sH \sigma I_\sH=I_\sH.$$
\end{enumerate}

\end{definition}

The Kubo--Ando theorem establishes a one-to-one correspondence between operator connections and non-negative operator monotone functions on $[0,\infty)$, see
\cite[Theorems 3.3, 3.4]{KA80}. The isomorphism is provided by $\sigma \mapsto f$, where $f(t)I_\sH:=I_\sH \sigma (tI_\sH)$. 
Its inverse $f \mapsto \sigma$ is defined by taking the integral expression \eqref{frep} of a non-negative operator monotone function $f$ on $[0,\infty)$, 
$f(t)=a+bt+\int_{(0,\infty)} \frac{(1+s) t}{t+s}d\mu(s)$, and then defining the corresponding $\sigma$ as   
\begin{equation}
\label{AsigmafB}
A\sigma B:= a A + b B + \int_{(0,\infty)}\frac{1+t}{t} [(tA):B] \, d\mu(t)
\end{equation}
Here, $A:B$ is the \emph{parallel sum} operator connection which is 
defined as the bounded quadratic form (see e.g. \cite[Lemma 3.1.5]{Hia10})
\begin{equation}
\label{C:D}
\langle \xi, (A:B)\xi \rangle:=\inf \{ \langle \zeta,A\zeta\rangle + \langle \xi-\zeta, B(\xi-\zeta) \rangle  \ : \ \zeta \in \sH\}.
\end{equation}
When $A$ and $B$ are positive operators with bounded inverses, then 
\begin{equation}
\label{A:B}
A:B = (A^{-1} + B^{-1})^{-1}.
\end{equation}

\begin{example}[Left and right trivial means]
     The left trivial mean $\sigma_1$ is induced by the function $f(x)\equiv 1$ and gives $A\sigma_1 B=A$. 
     The right trivial mean $\sigma_x$ is induced by $f(x)=x$ and gives 
     $A\sigma_x B=B$.
     \end{example}
     
     \begin{example}[$\alpha$--geometric means]
The $\alpha$--geometric means are defined in terms of
     the operator monotone function 
     \begin{equation}
     f_\alpha(t):=t^\alpha=\frac{\sin \alpha\pi}{\pi}\int_{0}^{\infty}\frac{t}{s + t}\frac{ds}{s^{1-\alpha}},\quad  t \ge 0, \alpha \in (0,1).
     \end{equation}
     The corresponding measure $\mu_\alpha$ and constants $a_\alpha, b_\alpha$ as in \eqref{frep} are therefore
     \begin{equation}
     d\mu_{f_\alpha} =\frac{\sin \alpha\pi}{\pi} \frac{t^\alpha dt}{t(t+1)}, \quad a_\alpha = b_\alpha = 0.
     \end{equation}
     Particular examples are the left- and right trivial mean and the geometric mean for which $\alpha = 1/2$. 
          \end{example}
 
  \begin{example}[Logarithm] The logarithm $f(t) = \log t$ is operator monotone on $(0,\infty)$ and 
  formally has $a=-\infty$, $b=0$ and $d\mu_{\log} = t^{-1}(1+t)^{-1}dt$.
  We will typically consider the approximation $f_n(t):=\log(t+\frac{1}{n}), n \in \mathbb{N}$ which is operator monotone on $[0,\infty)$ with $a_n= -\log n, b_n = 0$.
  \end{example}    
   
Consider two bounded positive operators such that $B \le \lambda A$ for some $\lambda < \infty$. Then $f(A^{-1/2} B A^{-1/2}) \le f(\lambda) I$ is 
a bounded operator and the Kubo-Ando mean $\sigma$ corresponding to the operator monotone 
function $f$ can also be expressed as \cite[Theorem 3.3]{KA80}
\begin{equation}
\label{AsigmafB1}
A \sigma B = A^{1/2} f(A^{-1/2} B A^{-1/2}) A^{1/2}, 
\end{equation}
as one may see using the integral representations \eqref{frep}, \eqref{AsigmafB} as well as the expression for parallel sum \eqref{A:B}. 
The following divergences first appeared in \cite{petz1998contraction} and were developed further in \cite{matsumoto2018new}, \cite{Hiai2017different}.

\begin{definition}\label{hiaidef0}
Consider a non-negative operator monotone function $f: [0,\infty) \to [0,\infty)$ characterized by \eqref{frep} and 
positive trace class operators $A,B$ such that $B \le \lambda A$ for some $\lambda > 0$. Then the
``maximal quantum $f$--divergence'' of $A$ with respect to $B$ is defined by
\begin{equation}
D_f(A\|B):=-\log {\rm Tr}_\sH (A \sigma B),
\end{equation}
where $\sigma$ is the operator connection corresponding to $f$.
\end{definition} 

\begin{remark}\label{hiaidef0rem}
For general positive trace class operators $A,B$ such that $B \le \lambda A$ possibly does not hold for any $\lambda < \infty$, one defines
\begin{equation}
D_f(A\|B) := \lim\limits_{\varepsilon \downarrow 0}D_f(A+ \varepsilon C \| B + \varepsilon C) \in (-\infty, +\infty].
\end{equation}
Here $C$ is any bounded positive operator such that $\lambda^{-1}(A+B) \le C \le \lambda(A+B)$ for some $\lambda < \infty$. 
By the monotonicity property of operator connections, the limit exists because the sequence is monotone decreasing. 
The limit is independent of the particular choice of $C$
as a special case of lemma \ref{lem:1} below. 
\end{remark}

\subsection{Maximal $f$-divergence for von Neumann algebras}

Operator connections and maximal $f$-divergences can be generalized from bounded operators 
to more general settings such as to suitable classes of
unbounded positive quadratic forms \cite{Kos17,Kos06,HK21}. In this work, we will mainly be interested in 
the notion of operator connection and maximal $f$-divergence between two positive normal functionals $\varphi, \psi$ on a von Neumann algebra $\cM$. 
This setting is investigated in great detail in \cite{Hia19,Hia21} to which we refer as general references. 
Based on the results \cite[Appendix D]{Hia21} one can for instance obtain quite easily 
a variational characterization of the maximal $f$-divergence which will be the 
basis of most developments in this work. 

The starting point in the von Neumann algebra setting is the Connes cocycle together with the following well-known result \cite{Con73}, see e.g.
\cite{araki1982positive} for the definitions of the modular operators $\Delta_\psi$ and Connes cocycles $[D\psi:D\varphi]_t$.

\begin{lemma}
 Let $\psi,\varphi$ be normal, positive functionals on $\cM$, assume that $\psi\leq \lambda \varphi$ for $\lambda >0$. Then the Connes cocycle derivative $[D\psi:D\varphi]_{t}$ admits an extension to a weakly continuous ($\cM$-valued) function $[D\psi:D\varphi]_z$ for $z$ in the strip $-1/2\leq\Im z\leq 0$ which is analytic in the interior. The generator $[D\psi:D\varphi]_{-i/2} \in \cM$ has norm less than $\sqrt{\lambda}$, and $\Delta_{\psi}^{1/2}=[D\psi:D\varphi]_{-i/2}\Delta_{\varphi}^{1/2}$.
\end{lemma}

Using this and the next lemma, one can define \cite{Hia19}: 

\begin{definition}\label{hiaidef}
Consider an operator monotone function $f: [0,\infty) \to [0,\infty)$, and normal, positive functionals $\varphi,\psi$ be on $\cM$ such that there exists $\lambda>0$ such that $\lambda^{-1}\varphi\leq \psi\leq \lambda\varphi$. Then we have the positive invertible operator $T_\varphi^\psi :=([D\psi:D\varphi]_{-i/2})^\ast [D\psi:D\varphi]_{-i/2} \in \cM$. Define the maximal quantum $f$--divergence of $\varphi$ with respect to $\psi$ by 
\begin{equation}
D_f(\varphi\|\psi):=-\log (\varphi(f(T_\varphi^\psi))).
\end{equation}
\end{definition} 

\begin{lemma}\label{lem:1}
(See \cite{Hia19})
Let $\varphi,\psi$ be normal, positive functionals on $\cM$. For every $\phi \sim \varphi+\psi$, i.e. there exists $\lambda>0$ such that $\lambda^{-1}(\varphi+\psi)\leq \phi\leq \lambda (\varphi+\psi)$, the limit $$\lim\limits_{\varepsilon \downarrow 0}D_f(\varphi + \varepsilon \phi \| \psi+\varepsilon  \phi) \in (-\infty, +\infty]$$ exists, and it is independent on the choice of $\phi$ above.
\end{lemma}

It therefore makes sense to make the following definition \cite{Hia19}.

\begin{definition}\label{epsdef}
Let $\varphi,\psi$ let normal, positive functionals on $\cM$, and let $f: [0,\infty) \to [0,\infty)$ be an operator monotone function. 
The maximal quantum $f$-divergence of $\varphi$ with respect to $\psi$ is defined as 
\begin{equation}
D_f(\varphi \| \psi):=\lim\limits_{\varepsilon \downarrow 0}D_f(\varphi + \varepsilon \phi \| \psi+\varepsilon  \phi)
\end{equation}
where $\phi$ is any positive normal functional on $\cM$ satisfying $\phi \sim \varphi + \psi$. 
Note that we may chose $\phi = \varphi + \psi$.
\end{definition}

\begin{remark}
If $\cM = B(\sH)$ is a type I von Neumann factor (or more generally, direct sum of factors), positive normal functionals on $\cM$ are in one to one correspondence with positive trace class operators
on $\sH$. Under this identification, the above definition of maximal $f$-divergence reduces to definition \ref{hiaidef0} and remark \ref{hiaidef0rem}. 
\end{remark}

\begin{remark}
The attentive reader will notice that compared to \cite{Hia19}, we require $f$ in definition \ref{hiaidef} to be operator monotone rather than
operator convex, the order of the states is reversed and we have a logarithm. The presence of the logarithm is just for convenience  
to make the entropy additive under tensor products. If we ignore the logarithm then our definition reduces to $ -\hat{S}_{-f}(\psi\|\varphi)$ of \cite{Hia19} 
for non-negative operator monotone functions noting that $-f$ is operator convex. 
Of course, the definition of \cite{Hia19} works for the larger class of all operator convex functions.
\end{remark}

\subsection{Properties of maximal $f$--divergence}

Many properties of $D_f$ and $D_{BS}$, and of the corresponding connections $\sigma$ between states\footnote{In this paper, we shall often use the somewhat imprecise 
terminology ``Kubo-Ando mean'' for this, even though a Kubo-Ando mean usually refers to the case when $f: [0,\infty) \to [0,\infty)$ is an operator monotone function 
with $f(1)=1$.}, 
in the setting of von Neumann algebras are known, see e.g. \cite[Theorem 4.4, Proposition 4.5]{Hia21}. In the present work, a variational formula for $D_f$ and the BS divergence $D_{BS}$ will take center stage and from this, 
many of these properties could be seen directly in retrospect. 
First, one defines an analogue of the parallel sum \eqref{A:B} for two normal positive functionals $\varphi,\psi$
on the von Neumann algebra $\cM$ by $(z \in \cM)$
\begin{equation}
\label{varphi:psi}
(\varphi : \psi)(zz^*) := \inf \{ \varphi(xx^*) + \psi(yy^*) \mid x+y = z, x,y \in \cM \}. 
\end{equation}
Then $\varphi : \psi$ is a positive normal functional on $\cM_+$ which is extended to 
all of $\cM$ by writing a general element as a difference of elements from $\cM_+$. 
Using the notion of parallel sum, one can next define a notion of operator mean $\varphi \sigma \psi$ associated with 
an operator monotone function $f: [0,\infty) \to [0,\infty)$ with representation \eqref{frep} between 
two positive normal functionals on $\cM$ by an analogue of the formula \eqref{AsigmafB}:
\begin{equation}
(\varphi \sigma \psi)(m) := a\varphi(m) + b\varphi(m) + \int_{(0,\infty)} \frac{1+t}{t} [(t \varphi):\psi](m) \, d\mu(t), 
\end{equation}
where $a,b,d\mu$ with $a,b<\infty$ correspond to the operator monotone function $f: [0,\infty) \to [0,\infty)$  as in \eqref{frep}. 
By combining \cite[Theorem D.7, D.8, D.10]{Hia21}, it follows that
\begin{equation}
D_f(\varphi \| \psi) = -\log \left(
a \varphi(1) + b \psi(1) + \int_{(0,\infty)} \frac{1+t}{t} [(t \varphi):\psi](1) \, d\mu(t)
\right). 
\end{equation}
From this relation and \eqref{varphi:psi} one can obtain the variational formula with ease, see also \cite[Remark 9.2]{hiai2022pusz} for 
a closely related formula in the case of $\cM = B(\sH)$:
\begin{proposition}
\label{varprop}
Let $f: [0,\infty) \to [0,\infty)$ be an monotone function with representation \eqref{frep} where $a,b<\infty$. 
Let $\psi,\varphi$ be normal, positive functionals on $\cM$. Then we have
\begin{equation}
\begin{split}
& D_f(\varphi\|\psi) \\
=& \ -\log \left( a\varphi(1) + b\psi(1) + \inf\limits_{(0,\infty)\xrightarrow{x} \cM} \int_{(0,\infty)}(1+t)\left[\varphi(x_t^{}x_t^\ast) + \frac{1}{t}\psi(y_t^{}y^\ast_t)\right]d\mu(t) \right),
\end{split}
\end{equation}
 where the infimum is taken over all step functions $x:(0,\infty)\rightarrow \cM$ with finite range 
 such that $x_t=1$ for sufficiently small $t$, such that $x_t=0$ for sufficiently large $t$, and where $y_t:=1-x_t$.
\end{proposition}
\begin{proof}
Note that since $0 \le [(t \varphi):\psi](1) \le {\rm min}(\psi(1), t\varphi(1))$ we can choose a $\delta>0$, and a $K<\infty$  such that 
\begin{itemize} 
\item $|\int_{(0,\delta)}(1+t)\varphi(1)d\mu-\int_{(0,\delta)} \frac{1+t}{t}
[(t \varphi):\psi](1) d\mu | \leq \varepsilon$
\item $|\int_{(K,\infty)}   \frac{1+t}{t}\psi(1)d\mu-\int_{(K,\infty)} \frac{1+t}{t}
[(t \varphi):\psi](1) d\mu| \leq \varepsilon $
\end{itemize}
Then we define $y_t=0, t<\delta$, and $y_t=1, t>K$. Next we build a step function $x: (\delta,K)\rightarrow \cM$ with finite range such 
that $\int_{(\delta,K)} (1+t)\left[\varphi(x_tx^\ast_t)+\frac{1}{t}\psi(y_ty_t^\ast)\right]d\mu$ 
approximates $ \int_{(\delta,K)}\frac{1+t}{t}
[(t \varphi):\psi](1) d\mu$ to within tolerance $\varepsilon$. 
This can be done by using the inner regularity of the Radon measure $\mu$; for details on this standard procedure see
\cite{Kos86}. 
\end{proof}

The above proposition does not cover the operator monotone function $f(t) = \log(t)$ (formally $a=-\infty$ in the 
representation \eqref{frep}). Since this case underlies the Belavkin-Staszewski (BS) divergence and is particularly interesting for us, we treat it explicitly. 
Consider first a pair of positive normal functionals on $\cM$ such that $\varphi \sim \psi$. 
Define the BS divergence as 
\begin{equation}
D_{BS}(\varphi\|\psi):=-\varphi(\log(T_\varphi^\psi))
\end{equation}
For a general pair of positive normal functionals such that $\varphi \sim \psi$ does not hold, 
we define $D_{BS}(\varphi\|\psi)$ analogously to definition \ref{epsdef}. The BS divergence can be seen as the limit $\alpha \to 1$ of 
the maximal geometric $\alpha$-divergence corresponding to $f_\alpha(t) = t^\alpha$, $\alpha \in (0,1)$. 

\begin{proposition}\label{BS}
We have
\begin{equation}
D_{BS}(\varphi\|\psi) = 
\sup \sup \left\{ \varphi(1)\log n - \int_{1/n}^{\infty} \left[\varphi(x^{}_t x^\ast_t) + \frac{1}{t}\psi(y^{}_ty_t^\ast)\right]\frac{dt}{t} \right\},
\end{equation}
where the first $\sup$ is taken over $n\in\mathbb{N}$, while the second is over finite range step functions $x$on $(\frac{1}{n},\infty)$
as in proposition \ref{varprop}.
\end{proposition}
\begin{proof}
Consider the approximating sequence $f_n(t):=\log(t+\frac{1}{n}), t\geq 0$, which has integral representations 
$f_n(t)=-\log n + \int_{1/n}^{\infty}\frac{s}{t+s}\frac{dt}{t}$. For every such $f_n$ the same argument of proposition \ref{varprop} applies. 
The claim follows by using the monotone convergence theorem on $(0,1)$ and the dominated convergence theorem on $[1,\infty)$.
\end{proof}

We now list some of the main properties of $D_f$ and $D_{BS}$, see \cite[Theorem 4.4, Proposition 4.5]{Hia21}. 
\begin{enumerate}
\item 
(Data processing inequality) 
Let $T: \cN \to \cM$ be a positive, normal, unital linear map between von Neumann algebras $\cM$ and $\cN$ satisfying the Schwarz property
$
T(nn^*) \ge T(n)T(n)^* 
$
for all $n \in \cN$\footnote{By Kadison's theorem, this property follows if $T$ is a 2-positive or even completely positive normal linear map.}. 
Then $D_f(\varphi \circ T\|\psi \circ T) \le D_f(\varphi \| \psi)$.

\item 
(Lower semi-continuity)
Let $\varphi_n, \psi_n$ be sequences of normal positive functionals converging pointwise to normal positive functionals $\varphi, \psi$ 
as $n \to \infty$. Then $D_f(\varphi \| \psi) \le \liminf_n D_f(\varphi_n \| \psi_n)$.

\item 
(Martingale property)
Consider a von Neumann algebra $\cM$ and an increasing sequence of 
von Neumann subalgebras $\{\cM_n\}$ such that $\cM=\left(\bigcup_n \cM_n \right)''$.  Then 
\begin{equation}
\label{invadingnet}
D_f(\varphi\restriction_{\cM_n}\|\psi\restriction_{\cM_n})\nearrow D_f(\varphi\|\psi).
\end{equation}

\item 
(Joint convexity and subadditivity)
The functional $\cM_{\ast,+}\times \cM_{\ast,+} \xrightarrow{D_f} (-\infty,+\infty]$ is jointly convex and subadditive.
\end{enumerate}

The analogous properties hold for $D_{BS}$.

\begin{remark}
Item 2 is not provided in \cite[Theorem 4.4, Proposition 4.5]{Hia21}; indeed it is presented as a conjecture for general von Neumann algebras \cite[Problem 4.13]{Hia21}. The variational expressions given in propositions \ref{BS} and \ref{varprop} provide an immediate proof of this for $D_f$ and $D_{BS}$, see for example \cite{Pet2} for the analogous argument for the (Araki) relative entropy.
\end{remark}

\section{Bimodules and $f$-divergences for channels}
\label{sec:chdiv}

\subsection{Definitions}

\begin{definition}[Bimodules]
Given von Neumann algebras $\cN, \cM$, a $\cN-\cM$ bimodule is a triple $\left(\sH,\ell_\sH, r_\sH \right)$ where $\sH$ is a Hilbert space, and 
$\cN \xrightarrow{\ell_\sH} B(\sH) \xleftarrow{r_\sH}\cM$ are a normal representation, and a normal anti-representation, respectively, such that $\ell_\sH(\cN)$ and $r_\sH(\cM)$ commute.
\end{definition}

When it is clear from the context, we will denote a bimodule by the underlying Hilbert space $\sH$. For more details on bimodules see e.g. \cite{Lon18,sauvageot1983produit}.

\begin{remark}
\label{standard bimodule}
We will use the natural notation $n\xi m := \ell_\sH(n)r_\sH(m)\xi, \ n\in \cN, m\in \cM,\xi \in \sH,$ when the bimodule Hilbert space is the identity bimodule $L^2(\cM,\Omega)$, which is the 
bimodule arising from a standard representation and standard vector $\Omega$ of $\cM$, so it is unique up to unitary equivalence. As a vector space, $L^2(\cM,\Omega)$ is realized up to unitary equivalence 
as the GNS Hilbert space $\sH$ of some chosen faithful normal state $\omega$ with associated cyclic and separating GNS vector $\Omega$. 
The right- and left action defining the $\cM-\cM$ bimodule structure of $L^2(\cM,\Omega)$ are defined as
\begin{equation}
\ell_\sH(m) \xi = m\xi, \quad r_\sH(m) \xi = Jm^*J\xi, 
\end{equation}
where $J=J_\Omega$ is the modular conjugation associated with $\Omega$ that sends $\cM$ anti-unitarily to $\cM'$.
\end{remark}

\begin{proposition}
\label{lprop}
Let $T:\cN\rightarrow\cM$ be a channel and let $\varphi$ be a normal state of $\cM$. There exists a $\cN-\cM$ bimodule $\sH_T$ and a vector 
$\xi \in \sH_T$ such that 
\begin{equation}
\label{longoeq}
\varphi(T(n)m)=\langle \xi,\ell_{\sH_T}(n)r_{\sH_T}(m)\xi \rangle,
\end{equation}
and $\xi$ is cyclic for $\ell_{\sH_T}(\cN)\lor r_{\sH_T}(\cM)$. 
Moreover, such a bimodule and unit vector are unique up to unitary transformations. 
\end{proposition}

\begin{proof}
This is \cite[Proposition 2.6]{Lon18}, but we report the proof here since it sets the stage for another proof below. 

Consider the map $\cN \times \cM^{op}\rightarrow \mathbb{C}: n\times m^{op}\mapsto \varphi(T(n)m^\ast)$. It is bilinear, so it corresponds to a linear map $\tilde{\varphi}:\cN \odot \cM^{op}\rightarrow \mathbb{C} $. Let $\xi$ be the vector representative of $\varphi$ in $L^2_+(\cM, \Omega)$ (the natural cone of the cyclic and separating vector $\Omega$), 
then $\tilde{\varphi}=\varphi_\xi\circ\pi\circ (T \otimes \text{id}),$ where $\pi$ is the representation of $\cM\odot \cM^{op}$ given by the $\cM-\cM$ bimodule
$L^2(\cM)$, see remark \eqref{standard bimodule}. Define $\sH_T$ as the $\cN-\cM$ bimodule given by the GNS representation of the positive 
functional $\tilde{\varphi}$ on $\cM\odot \cM^{op}$. Then clearly \eqref{longoeq} and the other assertions follow by the properties of the GNS representation. 
\end{proof}

In the following, $f: [0,\infty) \to [0,\infty)$ is an operator monotone function (such that $a,b<\infty$ in its representation \eqref{frep}).

\begin{definition}
\label{defchdiv}
Consider a pair of channels $S,T:\cN\rightarrow \cM$, and a von Neumann algebra $\cA$. We extend the channels to $S \otimes \text{id}_\cA, T \otimes \text{id}_\cA:\cN\odot \cA^{op} \rightarrow \cM \odot \cA^{op}$. Let $\pi$ be any binormal representation of  $\cM\odot \cA^{op}$. Then for every vector $\xi \in \sH_\pi$ in the $\cM-\cA$ bimodule given by $\pi$, we can consider the states $\varphi_{S,\pi,\xi}= \varphi_{\xi}\circ \pi \circ (S\otimes \text{id}_\cA)$ and $\varphi_{T,\pi,\xi}= \varphi_{\xi}\circ \pi \circ (T\otimes \text{id}_\cA)$. Consider $D_f(\varphi_{S,\pi,\xi}\|\varphi_{T,\pi,\xi})$
as defined by proposition \ref{BS}, now involving the supremum over finite range step functions $x$ with values in $\cN \odot \cA^{op}$. Then we define
\begin{equation}\label{channeldiv}
    D_f(S \| T):= \sup_{(\cA,\pi,\xi)} D_f(\varphi_{S,\pi,\xi}\| \varphi_{T,\pi,\xi}),
\end{equation}
where the supremum is over the triples $(\cA,\pi,\xi)$ consisting of a von Neumann algebra $\cA$, bimodule $\pi$ as above and normalized $\xi \in \sH$.
We make the analogous definition for the BS divergence.
\end{definition}

\begin{remark}
When $\cM$ is finite-dimensional, our definition for $D_{BS}$ agrees with that of \cite{Fan21}. This follows from proposition \ref{BS} and  the part of proposition
\ref{sbimodule} referring to finite dimensional type I algebras. In fact \cite{Fan21} also consider the channel divergence for the function $f(t) = t^\alpha, \alpha \in (1,2]$. 
That case is not considered in the present work since this function is not operator monotone but operator convex, and it is not obvious to what extent the variational formula 
in proposition \ref{varprop} still applies in this case. 
\end{remark}

\begin{remark}
Consider a normal homomorphism $\theta:\cA\rightarrow \cM$. Then a new bimodule $\sH_\theta$ can be constructed by \emph{twisting} the identity bimodule $L^2(\cM)$ on the right by using $\theta$. More explicitly, $\sH_\theta=L^2(\cM)$ as Hilbert space, the left action of $\cM$ is the one coming from the structure of $L^2(\cM)$ as a left $\cM-$module, while the right action of $\cA$ is defined $r_\theta(a)\eta:=\eta\theta(a), \ \eta \in L^2(\cM), \ a \in \cA$. In this case, the bimodule is denoted by $L^2_\theta(\cM,\Omega).$\footnote{Analogously, one may define ${}_\theta L^2(\cM)$.}
\end{remark}

Even though we will stick with the above definition in what follows, one 
may ask to what extent it is necessary to consider {\emph all} bimodules in the definition of the channel divergence \eqref{channeldiv}. 

\begin{proposition}
\label{sbimodule} 
If $\cM$ is properly infinite (direct sum of factors of types I$_\infty$, II$_\infty$ or III) or a direct sum of type I$_n$ factors then we have
\begin{equation}\label{standardsup}
 D_f(S\|T)=\sup\limits_{\xi \in L^2(\cM)_+} D_f(\varphi_{S,\pi,\xi}\|\varphi_{T,\pi,\xi}).
 \end{equation}
If $\cM$ is infinite dimensional and finite (direct sum of factors of type II$_1$) then 
\begin{equation*}
D_f(S\|T)=\sup\limits_{\xi \in \left(L^2(\cM) \otimes L^2(\ell^2\left(\mathbb{N}\right))\right)_+} D_f(\varphi_{S,\pi,\xi}\|\varphi_{T,\pi,\xi}).
\end{equation*} 
where by $L^2(\ell^2\left(\mathbb{N}\right))$ we mean the Hilbert-Schmidt operators on the separable Hilbert space $\ell^2(\mathbb{N})$
and by $L^2(\cM) \otimes L^2(\ell^2\left(\mathbb{N}\right))$ we mean the associated $\cM - \cM \otimes B(\ell^2(\mathbb{N}))$-bimodule.
The same holds for the BS divergence.
\end{proposition}

\begin{proof}
By definition $D_f(S\| T)\geq\sup\limits_{\xi \in L^2(\cM)_+} D_f(\varphi_{S,\pi,\xi}\|\varphi_{T,\pi,\xi}) $. 
To prove the reverse inequality, we can assume for the sake of simplicity that $\cN,\cM,\cA$ are all factors; the general 
case may be treated by performing the usual decomposition into a direct sum of factors. 

Case 1) $\cM$ is of type I$_\infty$, II$_\infty$, III. Then the sup in definition \ref{defchdiv} can always be realized for a properly infinite $\cA$ because we can take the tensor 
product $\cA \otimes B(\sH)$ and the corresponding bimodule if necessary. 
Consider a $\cM-\cA$ bimodule $\sH$. In this case, \cite[Corollary 2.7]{Lon18} implies that there exists a normal homomorphism $\theta:\cA\rightarrow \cM$, such that $\sH_\pi$ is isomorphic to $L^2_\theta(\cM)$. In other words, 
there exists a unitary $U:\sH\rightarrow L^2(\cM)$ intertwining the right representation of $\sH$ with the right representation of $L^2_\theta(\cM)$. For a vector $\xi\in \sH$, denote $\eta:=U\xi$. Then $D_f(\varphi_{S,\pi,\xi}\|\varphi_{T,\pi\,\xi})=D_f(\varphi_{S,\pi_\theta,\eta}\|\varphi_{T,\pi_\theta,\eta})$, where $\pi_\theta$ is the bimodule representation relative to $L^2_\theta(\cM)$. Now, using the variational formula in proposition \ref{varprop}
\begin{equation} 
\begin{split}
&D_f(\varphi_{S,\pi_\theta,\eta}\|\varphi_{T,\pi_\theta,\eta})\\
=&\sup\limits_{(0,\infty)\xrightarrow{x} \cN\odot\cA} -\log \Bigg( a\varphi_\eta((S\otimes\theta)(1))+b\varphi_\eta((T\otimes\theta)(1))+ \\  
&\int_{(0,\infty)}(1+t)\{\varphi_\eta((S\otimes\theta)(x_t^{}x_t^\ast))+\frac{1}{t}\varphi_\eta((T\otimes\theta)(y_t^*y_t^\ast)) \} d\mu \Bigg) \\
\leq& \sup\limits_{(0,\infty)\xrightarrow{v} \cN\odot\cM} -\ln \Bigg(a\varphi_\eta((S\otimes\text{id}_\cM)(1))+b\varphi_\eta((T\otimes\text{id}_\cM)(1))+ \\  
&\int_{(0,\infty)}(1+t)\{\varphi_\eta((S\otimes\text{id}_\cM)(v^{}_tv_t^\ast))+\frac{1}{t}\varphi_\eta((T\otimes\text{id}_\cM)(w^{}_tw_t^\ast)) \} d\mu \Bigg),
\end{split}
\end{equation}
where the inequality follows 
because the second sup is over a larger set which is easily seen by setting $v_t:=(\text{id}_\cN\otimes\theta)(x_t)$ and  $w_t=1-v_t$. The right side is 
$D_f(\varphi_{S, \pi, \eta}\|\varphi_{T,\pi,\eta})$, again by our variational formula. Taking the supremum over unit vectors $\eta$ in the natural cone
then demonstrates the reverse inequality $D_f(S\| T)\le \sup\limits_{\xi \in L^2(\cM)_+} D_f(\varphi_{S,\pi,\xi}\|\varphi_{T,\pi,\xi})$ and we are done.

Case 2) $\cM$ is of type I$_n$, i.e. $\cM = M_n(\mathbb{C})$. Let $(\xi,\pi,\cA)$ be a nearly optimal triple in definition \ref{defchdiv} with corresponding bimodule $\sH$, up to tolerance $\varepsilon$.  
By replacing $r_{\sH}(\cA)$ if necessary with the potentially larger von Neumann algebra $\ell_{\sH}(\cM)'$ (which is type I), we can assume that $r_{\sH}(\cA) = B(\sK)$, as well as 
$\sH =  \mathbb{C}^n \otimes \sK$. Now let $P$ be the orthogonal projection on $\sH$ with range $\ell_\sH(\cM)\xi$. Then $P \in \ell_{\sH}(\cM)'$ so 
$P = r_{\sH}(p)$ for some orthogonal projection $p \in \cA$, and by the Schmidt-decomposition theorem, $p$ has rank $\le n$. 
Going through 
the definitions, we then have for $n \in \cN, a \in \cA$:
\begin{equation}
\varphi_{S,\pi,\xi}(n \otimes a) = \langle \xi, S(n)\xi a\rangle = \langle \xi, S(n)\xi pap \rangle = \varphi_{S,\pi,\xi}((1_n \otimes p)(n \otimes a)(1_n \otimes p))
\end{equation}
Let $x$ be a step function valued in $\cM \odot \cA$ such that the infimum in proposition \ref{varprop} is achieved up to tolerance $\varepsilon$. 
Observe that
\begin{equation}
\varphi_{S,\pi,\xi}(x_tx_t^*) = \varphi_{S,\pi,\xi}((1_n \otimes p)x_t x_t^*(1_n \otimes p)) \ge  \varphi_{S,\pi,\xi}((1_n \otimes p)x_t (1_n \otimes p) x_t^*(1_n \otimes p)),
\end{equation}
and we get a similar relation for $S \to T$ and $x_t \to y_t = 1-x_t$. Therefore, setting $\hat x_t = (1_n \otimes p)x_t (1_n \otimes p), 
\hat y_t = (1_n \otimes p)y_t (1_n \otimes p)$, we have that 
\begin{equation} 
\label{split1}
\begin{split}
D_f(S\| T) -2\varepsilon \le& -\log \left(
a + b + \int_{(0,\infty)} (1+t) \{
\varphi_{S,\pi,\xi}(x_tx_t^*) + \frac{1}{t} \varphi_{T,\pi,\xi}(y_ty_t^*)
\} d\mu
\right) \\
\le & -\log \left(
a + b + \int_{(0,\infty)} (1+t) \{
\varphi_{S,\pi,\xi}(\hat x_t \hat x_t^*) + \frac{1}{t} \varphi_{T,\pi,\xi}(\hat y_t \hat y_t^*)
\} d\mu
\right) 
\end{split}
\end{equation}
Now observe that $1_n \otimes p$ is the unit in $\cM \odot p\cA p$ and that $\hat x_t + \hat y_t = 1_n \otimes p$, 
so $\hat x_t \in \cM \odot p\cA p$ is an admissible step function in the variational principle of proposition \ref{varprop}. Furthermore 
$p\cA p$ is naturally isomorphic to a subalgebra of $M_n(\mathbb{C}^n) = \cM$ (since the rank of $p$ is $\le n$), and that 
$P\sH$ (which contains $\xi = P\xi$) is isometric to a subspace of $\mathbb{C}^n \otimes \mathbb{C}^n$ (since $P = r_{\sH}(p)$). 
Therefore, the right side of \eqref{split1} is less than or equal to $D_f(\varphi_{S,\hat \pi, \hat \xi}\| \varphi_{T,\hat \pi, \hat \xi})$, 
where $\hat \pi$ is the representation associated with the standard $\cM-\cM$-bimodule $\mathbb{C}^n \otimes \mathbb{C}^n$, 
and where $\hat \xi$ is the unit vector in that bimodule corresponding to $\xi$. Thus, if $\cM = M_n(\mathbb{C}^n)$ it is sufficient 
in the variational definition \ref{defchdiv} to consider only $\cM-\cM$ bimodules.

Case 3) $\cM$ is of type II$_1$. Consider a $\cM-\cA$-bimodule $\sH$. Then  $\sH \otimes L^2(\ell^2(\mathbb{N}))$  is naturally a $\cM \otimes B(\ell^2(\mathbb{N})) - \cA \otimes B(\ell^2(\mathbb{N}))$-bimodule. Both algebras $\cM \otimes B(\ell^2(\mathbb{N}))$ and
$\cA \otimes B(\ell^2(\mathbb{N}))$ are properly infinite. Thus, by the same reasoning as above in 1), we know that the maximizing bimodule can be taken to be a standard bimodule for $\cM \otimes B(\ell^2(\mathbb{N}))$. Restricting this standard bimodule to a $\cM - \cM \otimes B(\ell^2(\mathbb{N}))$-bimodule gives a bimodule whose associated 
$f$-divergence as in definition \ref{defchdiv} is not smaller than that of the original $\cM-\cA$-bimodule. 

For the BS divergence, we consider the variational principle given by proposition \ref{BS} instead of proposition \ref{varprop}.
\end{proof}


\subsection{Basic properties of channel divergence}

We will now prove some basic properties of the channel divergences. In the next lemmas, $D_f$ is the divergence associated with a non-negative operator monotone function $f:[0,\infty) \to [0,\infty)$ with the representation \eqref{frep}
such that $a,b<\infty$ and $D_{BS}$ is the BS divergence (basically corresponding to $f(t) = \log t$). In our proof of the next theorem, we cannot use directly the 
proofs \cite[Theorem 4.4, Proposition 4.5]{Hia21} for $D_f(\varphi\|\psi)$ because our definition of $D_f(T\|S)$ is based fundamentally 
on a variational principle for testfunctions valued in $\cM \odot \cA^{op}$, which is not a von Neumann algebra.
Fortunately, it is well-known \cite{Kos86} (see also \cite[Chapter 5]{ohya2004quantum}) that variational  principles can provide alternative proofs.

\begin{theorem}
\label{thm1-4}
\begin{enumerate}
\item
(Lower semi-continuity)
Let $T_n, S_n$ be channels such that $T_n(m) \to T(m), S_n(m) \to S(m)$ weakly for any $m \in \cM$. Then 
$D_f(S \| T) \le \liminf_n D_f(S_n \| T_n)$, and similarly $D_{BS}(S \| T) \le \liminf_n D_{BS}(S_n \| T_n)$.

\item
(Data processing inequality)
Let $S_1, S_2: \cN \to \cM, T:\cR \to \cN$ be channels between von Neumann algebras. Then 
$D_f(S_1\| S_2) \ge D_f(S_1 \circ T\| S_2 \circ T)$. Similarly, 
 $D_{BS}(S_1\| S_2) \ge D_{BS}(S_1 \circ T\| S_2 \circ T)$. 
 
 \item
 (Joint convexity)
 Let $p_i, q_j$ be probability distributions over a finite set and $T_i, S_j:\cN \to \cM$ channels. Then 
 \begin{equation}
 D_{f}(\sum_i p_i S_i \| \sum_j q_j T_j) \le \sum_{i,j} p_i q_j D_{f}(S_i\|T_j),
 \end{equation}
 and similarly for $D_{BS}$.
 
 \item 
 (Dilation) Given channels $S,T:\cN\to\cM$, then $D_f(S\|T)=D_f(S\otimes id_{B(\ell^2(\mathbb{N}))}\|T\otimes id_{B(\ell^2(\mathbb{N}))})$. The same holds for $D_{BS}(S\|T)$.
\end{enumerate}
\end{theorem}

\begin{proof}
1) Consider an $\cM - \cA$ bimodule $\pi$ and unit vector $\xi \in \sH_\pi$ 
that achieves the supremum in the variational definition \eqref{channeldiv}
up to a tolerance $\varepsilon$. Then 
\begin{equation}
\begin{split}
D_f(S \| T) - 2\varepsilon =& D_f(\varphi_{S,\pi,\xi} \| \varphi_{T,\pi,\xi}) - \varepsilon\\
\le & \liminf_n D_f(\varphi_{S_n,\pi,\xi} \| \varphi_{T_n,\pi,\xi})\\
\le & \liminf_n D_f(S_n \| T_n).
\end{split}
\end{equation}
The first inequality is proven as follows. Consider an admissible step function $x: (0,\infty) \to \cM \odot \cA$ in proposition \ref{varprop} 
such that $D_f$ is achieved up to the tolerance $\varepsilon$, we see 
\begin{equation}
\begin{split}
& D_f(\varphi_{S,\pi,\xi} \| \varphi_{T,\pi,\xi}) -\varepsilon \\
=&-\log \left( a + b + \int_{(0,\infty)}(1+t)\{\varphi_{S,\pi,\xi}(x_t^{}x_t^\ast) + \frac{1}{t}\varphi_{T,\pi,\xi}(y_t^{}y^\ast_t)\}d\mu(t) \right) \\
=& -\log \left( a+b +\lim_n \int_{(0,\infty)}(1+t)\{\varphi_{S_n,\pi,\xi}(x_t^{}x_t^\ast) + \frac{1}{t}\varphi_{T_n,\pi,\xi}(y_t^{}y^\ast_t)\}d\mu(t) \right) \\
\le & \liminf_n D_f(\varphi_{S_n,\pi,\xi} \| \varphi_{T_n,\pi,\xi}),
\end{split}
\end{equation}
using again the variational principle in the last line. The statement for the BS divergence likewise follows from the corresponding variational
formula, see lemma \ref{BS}.

2) Let $(\xi, \pi, \cA)$ be a nearly optimal triple as in definition \eqref{channeldiv} for $D_f(S_1 \circ T\| S_2 \circ T)$ up to tolerance $\varepsilon$. 
Consider a step function $x_t \in \cN \otimes \cA^{op}$ as in proposition \ref{varprop} that is nearly optimal in the variational characterization 
of $D_f(\varphi_{S_1 \circ T,\pi,\xi} \| \varphi_{S_2 \circ T,\pi,\xi})$ up to a tolerance $\varepsilon > 0$. Then clearly $(T \otimes id)(x_t) \in \cM \odot \cA^{op}$ is 
 an admissible step function in the variational characterization 
of $D_f(\varphi_{S_1,\pi,\xi} \| \varphi_{S_2,\pi,\xi})$, and so we have, using Kadison's theorem 
$(T \otimes id)(n^*n) \ge (T \otimes id)(n)^*(T \otimes id)(n)$ and the unital property $T(1)=1$,
\begin{equation}
\begin{split}
& D_f(S_1 \circ T\| S_2 \circ T)- 2\varepsilon \le D_f(\varphi_{S_1 \circ T,\pi,\xi} \| \varphi_{S_2 \circ T,\pi,\xi}) - \varepsilon \\
\le& \ -\log \left( a + b + \int_{(0,\infty)}(1+t)\{\varphi_{S_1,\pi,\xi} \circ T(x_t^{}x_t^\ast) + \frac{1}{t}\varphi_{S_2,\pi,\xi} \circ T(y_t^{}y^\ast_t)\}d\mu(t) \right) \\
\le& \ -\log \left( a + b +  \int_{(0,\infty)}(1+t)\{\varphi_{S_1,\pi,\xi}(T(x_t^{})T(x_t^\ast)) + \frac{1}{t}\varphi_{S_2,\pi,\xi}(T(y_t^{})T(y^\ast_t))\}d\mu(t) \right)\\
\le& \ -\log \left( a + b + \inf\limits_{(0,\infty)\xrightarrow{x} \cM} \int_{(0,\infty)}(1+t)\{\varphi_{S_1,\pi,\xi}(x_t^{}x_t^\ast) 
+ \frac{1}{t}\varphi_{S_2,\pi,\xi}(y_t^{}y^\ast_t)\}d\mu(t) \right) \\
=& \ D_f(\varphi_{S_1,\pi,\xi} \| \varphi_{S_2,\pi,\xi}) 
\le \ D_f(S_1\|S_2),
\end{split}
\end{equation}  
and therefore the statement follows because $\varepsilon$ was arbitrary.
The proof for the BS divergence is similar and now based on proposition \ref{BS}.

3) The variational principles expressed in definition  \eqref{channeldiv} and proposition \ref{varprop} display $D_f(S\|T)$ as a double supremum of 
affine functionals of $S,T$. Joint convexity follows. The details are similar to 2).

4) Let $\tilde \sH$ be a $\cM \otimes B(\ell^2(\mathbb{N})) - \cA$ bimodule. Since a left representation is a right representation of the opposite algebra, 
this is also a $\cM - B(\ell^2(\mathbb{N}))^{op} \otimes \cA$ bimodule, hence included in the maximization in definition \ref{defchdiv} of $D_f(S\|T)$.
Thus, we have $D_f(S\|T) \ge D_f(S\otimes id_{B(\ell^2(\mathbb{N}))}\|T\otimes id_{B(\ell^2(\mathbb{N}))})$. On the other hand, 
let $(\pi,\cA,\xi)$ be a nearly optimal triple in the definition \ref{defchdiv} of $D_f(S\|T)$, up to tolerance $\varepsilon$, where $\pi$ corresponds to 
some $\cM-\cA$ bimodule $\sH$. Then $\tilde \sH = \ell^2(\mathbb{N}) \otimes \sH$ is a $\cM \otimes B(\ell^2(\mathbb{N}))-\cA$ bimodule. Let 
$\eta$ be any unit vector in $\ell^2(\mathbb{N})$ and set $\tilde \xi = \xi \otimes \eta$
as well as $\tilde S = S\otimes id_{B(\ell^2(\mathbb{N}))}$. If $x_t$ is a step function 
achieving the supremum the variational formula (proposition \ref{varprop}) of $D_f(\varphi_{\pi,S,\xi} \| \varphi_{\pi,S,\xi})$ up to tolerance $\varepsilon$, 
it follows that $\tilde x_t := x_t \otimes 1_{\ell^2(\mathbb{C})}$ is a valid step function in 
the variational definition of $D_f(\varphi_{\tilde \pi, \tilde S, \tilde \xi} \| \varphi_{\tilde \pi,\tilde S,\tilde \xi})$ achieving the same value. So we have 
$D_f(S\|T) -2\varepsilon \le D_f(S\otimes id_{B(\ell^2(\mathbb{N}))}\|T\otimes id_{B(\ell^2(\mathbb{N}))})$, and 4) is proven for $D_f$. In the case of $D_{BS}$
we use instead proposition \ref{BS}.
 \end{proof}
 
Our basic proof strategy to prove certain more profound properties of the channel divergence below 
will be to reduce the statements to those for finite dimensional matrix algebras obtained recently in \cite{Fan21}, and to do this 
we will need to restrict attention from now on to hyperfinite von Neumann algebras. 
[A von Neumann algebra $\cN$ on a Hilbert space $\sH$ is said to be 
hyperfinite if there exists a sequence (``filtration'') $\{\cN_n\}_{n\in \mathbb{N}}\subset \cN$ 
of finite dimensional subalgebras, increasing, i.e. $\cN_n \subset \cN_{n+1}$, such that $\cN=\left(\bigcup_n \cN_n \right)''$.] 

Examples of hyperfinite factors are all type I factors, i.e. $B(\sH), M_n(\mathbb{C})$, but there also exist hyperfinite factors of types II and III.
Let $\cN, \cM$ be hyperfinite factors, let $S,T: \cN \to \cM$ be channels, and let $\cN_n, \cM_n$ be filtrations of $\cN$, and $\cM$, respectively. Following \cite{Acc82} one 
can construct for each $n$ a ``generalized conditional expectation'' $E_n:\cM\rightarrow \cM_n$ as follows. Let $\omega_n = \omega \restriction \cM_n$
be the restriction of the faithful normal state $\omega$ on $\cM$, with GNS representation $\pi_n$, GNS Hilbert space $\sH_n$ 
and GNS vector $\Omega_n$. Let $J_n$ be the modular conjugation associated with $\Omega_n$ and define a partial isometry $V_n: \sH_n \to \sH$ by 
$V_n \pi_n(x)\Omega_n := x\Omega$ where $x \in \cM_n$. One can check that $V_n^* \cM' V_n \subset \pi_n(\cM_n)'$. It is therefore consistent to 
define $E_n(m)$ to be the unique element $x_n \in \cM_n$ such that
\begin{equation}
\pi_n(x_n) = J_n V_n^* JmJ V_n J_n, 
\end{equation}
where $J$ is the modular conjugation associated with $\Omega$ and $\cM$. By construction, each $E_n$ is a channel. 
 
 \begin{lemma}\label{weakconv} (see \cite{hiai1984strong})
  $E_n(m) \rightarrow m$ strongly as $n \to \infty$  for all $m \in \cM$.
 \end{lemma}

\noindent
We get: 
 
 \begin{proposition}[Martingale property] 
 We have $D_f(S\| T)=\sup_{n}  D_f(S \restriction_{\cN_n}\| T\restriction_{\cN_n})$. Similarly, 
 $D_{BS}(S\| T)=\sup_{n}  D_{BS}(S \restriction_{\cN_n}\| T\restriction_{\cN_n})$
 \end{proposition}
 
 \begin{proof}
Consider first the case $D_f$. 
Let $(\xi, \pi, \cA)$ be a nearly optimal triple as in the definition \eqref{channeldiv} of $D_f(S\|T)$. 
We let $p_n$ be the abstract units in $\cM_n$ which from an increasing net of projections in $\cM$. 
By the von Neumann density theorem, $p_n \to 1$ strongly. 
By monotonicity, i.e. by the data processing inequality applied to the inclusion channel $\cM_n \subset \cM$, we have
 $\limsup_n D_f(S\restriction_{\cM_n}\|T\restriction_{\cM_n})\leq D_f(S\|T)$.
 
Let $\cB := \bigcup_n \cM_n$ which is a $*$-subalgebra of $\cM$ whose weak closure is $\cB''=\cM$. By the von Neumann density theorem, $\cB$ is strongly dense on $\cM$. 
Let $x_t$ be an admissible step function as in proposition \ref{varprop} valued in $\cM \odot \cA^{op}$
which approximates $D_f(\varphi_{S,\pi,\xi} \| \varphi_{T,\pi,\xi})$ up to an arbitrary chosen tolerance $\varepsilon$. Because $x_t$
has finite range and because $\cB$ is strongly dense in $\cM$, we can construct a sequence  
the step functions $x_{n,t}$ in $\cM_n \odot \cA$ such that $x_{n,t}$ is constant on each interval where $x_t$ is constant, such that $x_{n,t} = p_n$ for any  $t$  which is so small that 
$x_t = 1$, and such that moreover $x_{n,t} \to x_t$ strongly on each such interval as $n \to \infty$. Then 
$\varphi_{S,\pi,\xi}(x_{n,t}^{}x_{n,t}^\ast) \to \varphi_{S,\pi,\xi}(x_{t}^{}x_{t}^\ast)$ and, letting $y_{n,t} = p_n - x_{n,t} \in \cM_n$, 
$\varphi_{T,\pi,\xi}(y_{n,t}^{}y_{n,t}^\ast) \to \varphi_{T,\pi,\xi}(y_{t}^{}y_{t}^\ast)$ as $n \to \infty$, uniformly in $t$. 
We insert the step functions $x_{n,t}, y_{n,t}$ and the unit $p_n$ instead of $x_t, y_t$ and $1$ 
into the right side of the variational formula proposition \ref{varprop}. 

The convergence properties of the step functions $x_{n,t}, y_{n,t}$ and the unit $p_n$ mean that
the right side converges to $D_f(\varphi_{S,\pi,\xi} \| \varphi_{T,\pi,\xi}) -\varepsilon$ as $n \to \infty$, 
thus $\liminf_n D_f(S\restriction_{\cM_n}\|T\restriction_{\cM_n})\ge D_f(\varphi_{S,\pi,\xi} \| \varphi_{T,\pi,\xi}) - \varepsilon \ge D_f(S\|T)-2\varepsilon$. 
We can take $\varepsilon$ smaller and smaller, proving that $\liminf_n D_f(S\restriction_{\cM_n}\|T\restriction_{\cM_n})\ge D_f(S\|T)$, which demonstrates the 
proposition. 

For the BS divergence, we proceed in a similar way now using variational principle proposition \ref{BS}. 
\end{proof}
 
 Combining the previous lemma and the martingale property we get:
  \begin{lemma}\label{sup1}
We have $D_f(S\| T)=\lim_m \sup_{n} D_f(E_m \circ S \restriction_{\cN_n}\| E_m \circ T\restriction_{\cN_n})$, and similarly for the 
BS divergence. 
 \end{lemma}

 \begin{proof}
 Consider the channels $E_m \circ S, E_m \circ T: \cN \to \cM_m$ and an $\cM_m - \cA$ bimodule $\sH_m$, representation $\pi_m$, and vector $\xi_m \in \sH_m$
as in the definition of $D_f(E_m \circ T\|E_m \circ S)$ such that the supremum \eqref{channeldiv} is achieved. 
By lemma \ref{sbimodule}, we may assume the bimodule in question to be the standard $\cM_m - \cM_m$ bimodule $L^2(\cM_m)$. 
From the channel $E_m: \cM \to \cM_m$ and the functional $\langle \xi_m, \ . \ \xi_m\rangle$, we then get an induced 
$\cM - \cM_m$ bimodule in view of proposition \ref{lprop}. It immediately follows that 
$D_f(S \| T) \ge D_f(E_m \circ S \| E_m \circ T)$ because in the variational definition \eqref{channeldiv}  of 
$D_f(S \| T)$, we take the supremum over the larger set of all bimodules whereas $D_f(E_m \circ S \| E_m \circ T)$ corresponds precisely 
to the induced bimodule $\cM - \cM_m$ just described. Thus, we see that $D_f(S \| T) \ge \limsup_m D_f(E_m \circ S \| E_m \circ T)$, whereas
$D_f(S \| T) \le \liminf_m D_f(E_m \circ S \| E_m \circ T)$ follows in view of the lower semi-continuity of the channel divergence because 
$E_m$ is pointwise strongly -- hence weakly -- convergent by lemma \ref{weakconv}. 
Therefore, we see that, simply $D_f(S \| T)=  \lim_m D_f(E_m \circ S \| E_m \circ T)$. The statement now follows from the martingale property.
The proof for the BS divergence is similar and based instead on proposition \ref{BS}.
\end{proof}
 
The following property, observed and proven first in \cite{Fan21} for matrix algebras, is crucial for this work.
 
 \begin{proposition}[Internal subadditivity]
 Let $S_2, T_2: \cN \to \cR, S_1, T_1: \cR \to \cM$ be channels between hyperfinite von Neumann algebras. Then we have
\begin{equation}
D_{BS}(S_2 \circ S_1 \| T_2 \circ T_1) \leq \sum\limits_{i=1,2} D_{BS}(S_i\|T_i).
\end{equation}
\end{proposition}
\begin{proof}
Let $E_m: \cR \to \cR_m, F_k:\cM \to \cM_k$ be sequences of generalized conditional expectations as described above. 
\begin{equation}
\begin{split}
& D_{BS}(S_1 \circ S_2 \| T_1 \circ T_2) \\
\le& \liminf_m D_{BS}(S_1 \circ E_m \circ S_2 \| T_1 \circ E_m \circ T_2)\\
\le& \liminf_m \liminf_k D_{BS}(F_k \circ S_1 \circ E_m \circ S_2 \| F_k \circ T_1 \circ E_m \circ T_2)\\
=& \liminf_m \liminf_k \sup_n D_{BS}(F_k \circ S_1 \circ E_m \circ S_2 \restriction \cN_n \| F_k \circ T_1 \circ E_m \circ T_2 \restriction \cN_n )\\
\le& \liminf_m \liminf_k \sup_n \Big( 
D_{BS}(F_k \circ S_1  \restriction \cR_n \| F_k \circ T_1 \restriction \cR_n ) \ + \\
& \hspace{4cm}
D_{BS}(E_m \circ S_2 \restriction \cN_n \|E_m \circ T_2 \restriction \cN_n ) 
\Big)
\\
=& D_{BS}(S_1\| T_1) + D_{BS}(S_2\|T_2)
\end{split}
\end{equation}
In the first two lines we used lower semi-continuity, in the third line we used the martingale property, in the fourth line we used 
the result by \cite{Fan21} in the context of finite-dimensional von Neumann algebras, and in the last step we used 
the martingale property and lemma \ref{sup1}. 
\end{proof}

\begin{remark}
A noteworthy special case of the proposition arises when $\cM = {\mathbb C}$, i.e. $S_1, T_1$ are states. In this case 
the subadditivity corresponds to the ``chain rule'' of \cite{Fan21}.
\end{remark}

Consider channels $S_i, T_i: \cN_i \to \cM_i$ between hyperfinite von Neumann algebras represented on Hilbert spaces $\sH_i$, where $i=1,2$. We can form 
the weak closure of $\cN_1 \odot \cN_2$ in $\cB(\sH_1 \otimes \sH_2)$ and denote this (hyperfinite) von Neumann algebra by $\cN_1 \bar \otimes \cN_2$, 
and we proceed similarly for $\cM_i$.  Then it follows that $S_1 \otimes S_2: \cN_1 \odot \cN_2 \to \cM_1 \otimes \cM_2$ can be extended to 
a channel $S_1 \otimes S_2$ from $\cN_1 \bar \otimes \cN_2 \to \cM_1 \bar \otimes \cM_2$, and similarly for $T_1 \otimes T_2$. Then we have: 

 \begin{proposition}[External additivity] Let $S_i, T_i: \cN_i \to \cM_i$ be channels between the hyperfinite von Neumann algebras $\cN_i, \cM_i, i=1,2$. Then
 $D_{BS}(S_1\otimes S_2\| T_1\otimes T_2) = \sum\limits_{i=1,2} D_{BS}(S_i\|T_i)$.
 \end{proposition}

\begin{proof}
Similar to the proof of internal subadditivity using again lemma \ref{sup1} and that $D_{BS}$ is additive under the tensor product in the finite dimensional case by 
results of  \cite{Fan21}.
\end{proof}
\subsection{Channel divergences for Kraus channels}

Let $\cM$ be a von Neumann algebra in standard form acting on the Hilbert space $\sH$ with cyclic and separating vector $\Omega$. 
We consider a class of channels $T,S: \cM \to \cM$ of so-called ``Kraus type'' investigated in the context of general von Neumann algebras 
by \cite{stormer2015analogue}. By definition, these are of the form 
\begin{equation}
\label{Krausch}
\begin{split}
&S(m) = \sum_{i=1}^N a_i^* m a_i, \quad  \sum_{i=1}^N a_i^*a_i^{} = 1,  \\
&T(m) = \sum_{i=1}^M b_i^* m b_i, \quad  \sum_{i=1}^M b_i^*b_i^{} = 1,  
\end{split}
\end{equation}
with $m,a_i,b_j \in \cM$ and $N,M \in {\mathbb N}$. Our aim is to give a formula for 
$D_{BS}(S \|T)$ for the channel divergence of two Kraus channels in terms of their ``Choi operators'' also introduced in this context by \cite{stormer2015analogue}. 
To this end, we define the Choi operators $C_S, C_T \in B(\sH)$ for such channels as, respectively
\begin{equation}
C_S = \sum_{i=1}^N a_i^{}|\Omega \rangle \langle \Omega | a_i^*, \quad 
C_T = \sum_{i=1}^M b_i^{}|\Omega \rangle \langle \Omega | b_i^*
\end{equation}
By construction $C_S \in B(\sH)$ is a non-negative operator of finite rank such that ${\rm Tr}_\sH C_S = \sum_{i=1}^N \| a_i \Omega \|^2 = 1$, and similarly for $C_T$.

Let $\cC \subset B(\sH)$ be the $*$-subalgebra of all operators of the form $\sum_{i=1}^{N} c_i |\Omega \rangle \langle \Omega | d_i$ for some $N \in {\mathbb N}, c_j, d_j \in \cM$. 
By \cite[Theorem 4]{stormer2015analogue}, the spectral projections of the operators $C_S, C_T$ are in $\cC$ and consequently this algebra is closed under the spectral calculus. Now suppose $\sigma$ is the Kubo-Ando connection associated with a non-negative operator monotone function $f:[0,\infty) \to [0,\infty)$ with $a,b<\infty$ in \eqref{frep}. It follows from \eqref{AsigmafB1} that the expressions under the limit in
\begin{equation}
C_S \sigma C_T = \lim_{\varepsilon \to 0} \Big(C_S + \varepsilon (C_S+C_T) \Big)\sigma\Big(C_T + \varepsilon (C_S+C_T) \Big) 
\end{equation}
are non-negative elements from $\cC$. Since the arguments of the mean $\sigma$ are decreasing and strongly convergent as $\varepsilon \to 0$, the limit not only exists 
by the properties of the Kubo-Ando connections, but is also in $C_S \sigma C_T \in \cC$, see the proof of \cite[Theorem 4]{stormer2015analogue}. Hence $C_S \sigma C_T$ is 
in particular a non-negative finite rank operator in $B(\sH)$. 

For the operator monotone function $f(t) = \log t$ on $(0,\infty)$, similar arguments show that the operators $[C_S + \varepsilon (C_S+C_T)]\sigma[C_T + \varepsilon (C_S+C_T)]$
are still in $\cC$ for $\varepsilon>0$. The limit $\varepsilon \to 0$ of this decreasing sequence exists but possibly only in the sense of an unbounded quadratic form.  
In fact, as long as $\varepsilon>0$, one can see e.g. from \eqref{AsigmafB1}, $\|C_S\|, \|C_T\| \le 1$ together with
$V^* f(A) V \le f(V^*AV)$ for contractions $V$, positive $A \in B(\sH)_+$ and operator monotone functions $f: (0,\infty) \to \mathbb R$ (see e.g. \cite{Pet2}), that
\begin{equation}
\begin{split}
& [C_S + \varepsilon (C_S+C_T)]\sigma[C_T + \varepsilon (C_S+C_T)] \\
\le & \max\{1, (1+\varepsilon)\| C_S \| + \varepsilon \|C_T\|\} \log[C_T + \varepsilon (C_S+C_T)] \\
\le & (1+2\varepsilon) \log (1+2\varepsilon) \ 
\xrightarrow{\varepsilon \to 0}  \ 0. 
\end{split}
\end{equation}
Thus, for $f(t)=\log t$, the corresponding Kubo-Ando mean $C_S \sigma C_T$ 
defines a negative possibly unbounded quadratic form given by a finite rank operator in $\cC$ on its domain. 
Assume that $C_S \sigma C_T$ is bounded, hence in $\cC$. Then the positive finite rank operator $-C_S \sigma C_T$ may be written as a linear combination of its eigenprojections as 
$\sum_{j=1}^K c_j |\Omega\rangle \langle \Omega |c_j^*$ for some $c_j \in \cM, K \in \mathbb{N}_0$, 
which gives, for $m' \in \cM'$ 
\begin{equation}
\label{OmegaST}
\langle \Omega_{S,T}, m'{}^*m'\Omega_{S,T}\rangle = \sum_{j=1}^K \langle \Omega, m'{}^*c_j c_j^* m'\Omega \rangle \le 
\left( \sum_{j=1}^K \| c_j\|^2 \right) \langle \Omega, m'{}^*m'\Omega\rangle.
\end{equation}

\begin{definition}
Let $\sigma$ be the Kubo-Ando mean for $f(t) = \log t$ and assume that $C_S \sigma C_T$ is bounded (hence in $\cC$). 
Then we define $\Omega_{S,T} \in L^2(\cM,\Omega)_+$ as the unique representer of the positive normal functional on $\cM'$ associated with the non-negative 
finite rank operator $-C_S \sigma C_T$ for the operator mean associated with $f(t) = \log t$,
\begin{equation}
\label{OmST}
\langle \Omega_{S,T}, m'\Omega_{S,T}\rangle = -{\rm Tr}_\sH \Big[ m'(C_S \sigma C_T) \Big], \quad m' \in \cM'.
\end{equation}
\end{definition}

\begin{remark}
By the Connes-Radon-Nikodym theorem and \eqref{OmegaST}, there must  be $m \in \cM$ such that $\Omega_{S,T} = m\Omega \in \cM\Omega$. We therefore have
$\Omega_{S,T} \in L^\infty(\cM,\Omega) \cong \cM\Omega$ by the well-known characterization of this space. 
\end{remark}

\begin{proposition}
\label{Krausprop}
For two Kraus channels $S,T$ on the finite dimensional or properly infinite hyperfinite von Neumann algebra $\cM$ standardly
represented on $L^2(\cM, \Omega)$ we have 
\begin{equation}
D_{BS}(S \| T) =  \Big\| \Omega_{S,T} \Big\|_{L^\infty(\cM, \Omega)}^2
\end{equation}
with the convention that the right side is $+\infty$ if $C_S \sigma C_T$ is unbounded.
\end{proposition}

\begin{proof}
First assume that $C_S \sigma C_T$ is bounded, so $\Omega_{S,T} \in L^\infty(\cM,\Omega) \cong \cM\Omega$ by the preceding remark.
By proposition \ref{sbimodule} we can restrict attention to the standard bimodule $\sH=L^2(\cM, \Omega)$ in the variational definition \eqref{channeldiv} of channel divergence. 
Furthermore, since $\cM'\Omega$ is strongly dense in $\sH$ as $\Omega$ is standard, it is sufficient to restrict to vectors $\xi \in \sH$ of the form $\xi = x'\Omega, x' \in \cM'$ in the variational 
definition. We get using the definitions and the notations $m' = Jm^{op *} J \in \cM'$, $x' = Jx^{op *} J$ and $X:=\pi(1 \otimes x^{op})$,
\begin{equation}
\label{41}
\begin{split}
    \varphi_{S,\pi,\xi}(m \otimes m^{op}) &= 
    \langle \xi, \pi (S(m) \otimes m^{op})\xi\rangle \\
    &= \sum_{i=1}^N \langle \xi , \ell_{\sH}(a_i^* m a_i^{}) r_{\sH}(m^{op})  \xi \rangle \\
    &= \sum_{i=1}^N \langle x'\Omega , a_i^* m a_i^{} m' x'  \Omega \rangle \\
    &= \sum_{i=1}^N \langle  x' a_i\Omega ,   m m' x' a_i \Omega \rangle \\
    &= {\rm Tr}_{\sH}\Big[XC_SX^*
    \pi(m \otimes m^{op})\Big]
\end{split}
\end{equation}
We also have a similar formula replacing $S$ by $T$ and $a_j$ by $b_j$. The variational principle for the maximal $BS$-divergence (proposition \ref{BS})
thereby gives us 
\begin{equation}
\begin{split}
&D_{BS}(\varphi_{S,\pi,\xi} \| \varphi_{T,\pi,\xi}) \\
=& \sup\sup \left( \log n
- \int_{1/n}^\infty
\{
\varphi_{S,\pi,\xi}(v_t^{}v_t^*) + t^{-1} 
\varphi_{T,\pi,\xi}(w_t^{}w_t^*)
\} \frac{dt}{t} \right) \\
=& \sup\sup \left( \log n
- \int_{1/n}^\infty
\{
{\rm Tr}_\sH(V_t^*XC_SX^*V_t^{}) + t^{-1} 
{\rm Tr}_\sH(W_t^*XC_TX^*W_t^{})
\} \frac{dt}{t} \right)
\end{split}
\end{equation}
by \eqref{41}, where the first supremum is over $n \in \mathbb N$, 
the second supremum is over the finite range step functions $(1/n,\infty) \xrightarrow{v} \cM\odot\cM^{op}$ such $v_t = 0$ for sufficiently large $t$, and 
where we use the abbreviations $V_t = \pi(v_t), w_t = 1-v_t, W_t = \pi(v_t)$. Since the strong closure of $\pi(\cM \odot \cM^{op})$ is strongly dense in $B(\sH)$, the 
step functions $V_t$ can be used to approximate in the strong topology any given finite range step function  
$(0,\infty) \to B(\sH)$ which is zero for sufficiently large $t$ and $1$ for sufficiently small $t$. Let $P$ be any orthogonal projection onto 
a finite dimensional subspace of $\sH$ containing the (finite dimensional) ranges of $XC_SX^*$ and $XC_TX^*$. Then it follows that we may further replace 
$V_t$ by $PV_tP$ and $W_t$ by $PW_tP$ and the variational formula \cite[Remark 9.2]{hiai2022pusz} (or our proposition \ref{BS})
 therefore tells us that 
\begin{equation}\label{previous}
\begin{split}
D_{BS}(\varphi_{S,\pi,\xi} \| \varphi_{T,\pi,\xi})
=& - {\rm Tr}_\sH \Big[ (XC_SX^*) \sigma (XC_TX^*) P \Big]\\
=& - {\rm Tr}_\sH \Big[ (XC_SX^*) \sigma (XC_TX^*) \Big]\\
=& - {\rm Tr}_\sH \Big[ X(C_S \sigma C_T)X^* \Big] \\
=& - {\rm Tr}_\sH \Big[ \pi(1 \otimes x^{op}) (C_S \sigma C_T)\pi(1 \otimes x^{op})^* \Big] \\
=& - {\rm Tr}_\sH \Big[ x'{}^* x' (C_S \sigma C_T) \Big] \\
=& \Big\| x' \Omega_{S,T} \Big\|^2,
\end{split}
\end{equation}
where we used that $P$ was arbitrary so long as its range ranges of $XC_SX^*$ and $XC_TX^*$ to go the third line, and
where we used the transformer equality (see e.g. \cite[Lemma D.3]{Hia21}) to go to the fourth line. The last step is admissible if we assume that $x'$, hence $X$, is invertible, which we assume momentarily is the case. 
Since we know that $\Omega_{S,T} \in L^\infty(\cM,\Omega)$, there is $m \in \cM$ such that $\Omega_{S,T} = m\Omega$, therefore 
\begin{equation}
\label{DBS}
D_{BS}(\varphi_{S,\pi,\xi} \| \varphi_{T,\pi,\xi}) = \Big\| x' m \Omega \Big\|^2 =\Big\| m\xi \Big\|^2.
\end{equation}
If we could show that $x'\Omega$ with $x' \in \cM'$ ranging over the invertible elements is dense in $\sH$, 
then this formula would hold on for all $\xi \in \sH$. This follows, in fact, from the hyperfinite 
property because invertible elements are norm dense in a finite-dimensional von Neumann algebra, and $\cM$ is the strong closure of hyperfinite algebras. Thus, we 
get a strongly convergent sequence $x_n \to x$ with $x_n$ invertible for any $x \in \cM$. Applying this to $x:= Jx'J$ and choosing $x_n' = Jx_nJ$ gives the statement. 
Taking the supremum over our strongly dense set of vectors $\xi$ with unit norm now gives the statement of the proposition because $\|m\| = \| \Omega_{S,T}\|_{L^\infty(\cM,\Omega)}$.

Let us now assume  that $C_S \sigma C_T$ is not bounded. The completely positive maps $T_\varepsilon:= T+\varepsilon(S+T), S_\varepsilon := S+\varepsilon(S+T)$ do not 
suffer from this problem for $\varepsilon >0$ and are (non-normalized) increasing (as $\varepsilon \to 0$) sequences 
of Kraus channels. By monotonicity of the operator mean $\sigma$,  $C_{S_\varepsilon} \sigma C_{T_{\varepsilon}}$ is an increasing sequence of 
self-adjoint operators in $\cC$ whose range remains in a fixed finite dimensional subspace of $\cH$. Hence it is convergent to the unbounded 
operator $C_S \sigma C_T$ in norm from which we can see that there must be $x' \in \cM$ such that $- {\rm Tr}_\sH [ x'{}^* x' (C_{S_\varepsilon} \sigma C_{T_{\varepsilon}}) ]$
diverges to $+\infty$, hence so does $D_{BS}(S_\varepsilon \| T_\varepsilon)$ by \eqref{previous}. 
However, since $T_\varepsilon, S_\varepsilon$ are decreasing sequences of channels, by monotonicity $D_{BS}(S \| T) \ge D_{BS}(S_\varepsilon \| T_\varepsilon) \to \infty$. 
\end{proof}

\subsection{Examples}

As a simple special case of Kraus channels we consider $S,T$ in \eqref{Krausch} of the form
\begin{equation}
\label{Cuntz}
\sum_{i=1}^N a_i^{} a_i^* = 1, \quad 
a_i^* a_j = \delta_{i,j}1, \quad \sum_{i=1}^M b_i^{} b_i^* = 1, \quad  b_i^* b_j = \delta_{i,j}1
\end{equation}
where $\{a_j\}$ respectively $\{b_j\}$ each generate algebras isomorphic to the Cuntz algebras on $N$ respectively $M$ isometries. 

\begin{corollary}
\label{cor:1}
Let $\cM$ be a finite dimensional or properly infinite von Neumann algebra and let $S,T$ be 
Kraus channels such that \eqref{Cuntz} holds for some $N,M \in \mathbb N$. 
Then either $N=M$ and $T=S$, or we have $D_{BS}(S\|T) = \infty$, or $N<M$ and $D_{BS}(S\|T)=0$.
\end{corollary}

\begin{remark}
In particular, note that if $T(m) = u m u^*$ with $u \in \cM$ unitary 
and $S=id$, we have $D_{BS}(id \| T) = D_{BS}(T \| id) = \infty$
unless $u=\lambda 1$. 
\end{remark}

\begin{proof}
It follows from the Cuntz algebra relations that the corresponding Choi operators are $C_T = Q$ and $C_S = P$ are orthogonal projections of rank $N$ respectively $M$ on 
$\sH$. Denote by $P \wedge Q$ the orthogonal projection onto the intersection of the ranges of $P$ and $Q$. Consider the operator monotone functions 
$f_n(t):=\log(t+\frac{1}{n}), t\geq 0$, which have integral representations $f_n(t)=-\log n + \int_{1/n}^{\infty}\frac{s}{t+s}\frac{dt}{t}$, 
and let $\sigma_n$ be the corresponding operator means. By the proof of \cite[Theorem 3.7]{KA80}, 
we have $(tP):Q = t(t+1)^{-1} (P \wedge Q)$. By the integral representation \eqref{AsigmafB} for this mean, we therefore get 
\begin{equation}
\begin{split}
P \sigma_n Q 
=& P \log \tfrac{1}{n} + \int_{1/n}^\infty [(tP):Q] \frac{dt}{t^2}\\
=& P \log \tfrac{1}{n} + (P \wedge Q) \int_{1/n}^\infty \frac{dt}{t(t+1)}\\
=& [(P \wedge Q)-P]\log n - (P \wedge Q)\Bigg[ \log n - (\log t - \log(1+t))_{1/n}^\infty \Bigg]\\
=&  [(P \wedge Q)-P]\log n + (P \wedge Q) \log (1+\tfrac{1}{n}) 
\end{split}
\end{equation}
As $n \to \infty$, the operator means $P\sigma_n Q$ are decreasing (hence convergent) to the potentially unbounded quadratic form $P\sigma Q = [(P \wedge Q)-P]\infty$, 
where $\sigma$ corresponds to the operator monotone function $f(t) = \log t$. 
Therefore, if $P \sigma Q$ is to be bounded, we must have $(P \wedge Q)-P = 0$, so $P$ must be a subprojection of $Q$, or otherwise $D_{BS}(S\| T) = \infty$ by  proposition \ref{Krausprop}. 
If $P=Q$, then $N=M$, and there must be $R_{ij} \in \mathbb C$
such that $a_i \Omega = \sum_{j=1}^N R_{ij} b_j \Omega$, and since $\Omega$ is separating, we must 
have $a_i = \sum_{j=1}^N R_{ij} b_j$. The Cuntz algebra relations then show that $(R_{ij})$ is a unitary matrix and 
then clearly $S=T$. If $P<Q$, then clearly $N<M$ and it follows that $P \sigma_n Q$ is decreasing (hence convergent) to $0$ and $D_{BS}(S\|T) = 0$.
\end{proof}

Another very simple but conceptually relevant example is:

\begin{proposition}
Let $\cM$ be a finite dimensional or properly infinite, hyperfinite von Neumann algebra 
and let $e_j \in \cM$ be $N$ mutually orthogonal projections such that $\sum_i e_i=1, 0<e_i<1$. We consider the Kraus channel 
\begin{equation}
M(m) = e_1me_1 + \dots + e_N m e_N
\end{equation}
corresponding to an $N$-ary measurement. Then 
\begin{equation}
D_{BS}(id\| M) = \log N
\end{equation}
\end{proposition}
\begin{proof}
The Choi operator associated with $M$ is $C_M = \sum e_i|\Omega\rangle \langle \Omega| e_i$ and that 
for the identity channel $id$ is $C_{id} = |\Omega\rangle \langle \Omega|$. We begin by working out the parallel sum 
$\langle \xi, [(tC_{id}):C_M] \xi \rangle$ using the variational definition \eqref{C:D}. A minimizer $\zeta_0$ in that definition has to satisfy 
\begin{equation}
t\langle \Omega, \zeta_0\rangle \Omega - \sum_i \langle e_i\Omega, \xi-\zeta_0\rangle e_i\Omega  = 0.
\end{equation}
The vectors $e_i\Omega$ are non-zero and linearly independent because $\Omega$ is separating and because the $e_i$'s are orthogonal and non-trivial. We therefore 
see that 
\begin{equation}
t\langle \Omega, \zeta_0\rangle = \langle e_i\Omega, \xi-\zeta_0\rangle 
\end{equation}
for all $i=1, \dots, N$ and any solution $\zeta_0$ is a minimizer for the variational problem \eqref{C:D}.
To find a solution we consider the ansatz
$\zeta_0 = \sum_i a_i \| e_i \Omega\|^{-2 }e_i\Omega$, leading to a linear system for the unknown complex coefficients $a_i$. 
A solution is 
\begin{equation}
a_i = \langle e_i \Omega, \xi\rangle - \frac{t}{1+Nt} \langle \Omega, \xi\rangle.
\end{equation}
Substituting the corresponding $\zeta_0$ into the variational definition \eqref{C:D} yields
\begin{equation}
\langle \xi, [(tC_{id}):C_M] \xi \rangle  = \frac{t}{Nt+1} | \langle \xi, \Omega \rangle |^2, 
\end{equation}
noting that the dependence upon $e_i$ has cancelled. In other words $[(tC_{id}):C_M] = \tfrac{t}{Nt+1}|\Omega \rangle \langle \Omega|$. 
Next we use the integral representation \eqref{AsigmafB} for the Kubo-Ando means $\sigma_n$ associated with the functions $f_n(t) = \log(\tfrac{1}{n}+t)$. 
The corresponding measures $d\mu_n$ are read off from the  integral representations $f_n(t)=-\log n + \int_{1/n}^{\infty}\frac{s}{t+s}\frac{dt}{t}$.
This gives for the Kubo-Ando mean $C_{id} \sigma C_M$ associated with $f(t) = \log t$ as required for the BS divergence,
\begin{equation}
\begin{split}
-C_{id} \sigma C_M =& \lim_n -C_{id} \sigma_n C_M \\
=&\lim_n \left( (\log n)C_{id} - \int_{(1/n,\infty)} [(tC_{id}):C_M] \frac{dt}{t^2} \right)\\
=& \lim_n \lim_K \left( (\log n) |\Omega \rangle \langle \Omega| - \int_{(1/n,K)}\left( \frac{t}{Nt+1}  |\Omega \rangle \langle \Omega| \right) \frac{dt}{t^2} \right) \\
=&\lim_n \lim_K \left( \log n - \int_{(1/n,K)} \frac{dt}{t(Nt+1)}   \right)  |\Omega \rangle \langle \Omega| \\
=& \lim_n \lim_K \left( \log n - \log (t) \bigg|_{1/n}^K + \log(Nt+1) \bigg|_{1/n}^K \right)  |\Omega \rangle \langle \Omega| \\
=&(\log N)|\Omega \rangle \langle \Omega|.
\end{split}
\end{equation}
Next we use the definition \eqref{OmST} for $S=id, T=M$, giving 
\begin{equation}
\Omega_{id,M} = (\log N)^{1/2} \Omega. 
\end{equation}
By proposition \ref{Krausprop}, we therefore have 
$
D_{BS}(id\|M) = \log N
$ 
as we wanted to show.
\end{proof}

Our final example concerns finite index inclusions of von Neumann factors. 

\begin{proposition}
\label{prop:E}
Let $\cN \subset \cM$ be a finite index inclusion of von Neumann factors with associated minimal conditional expectation 
$E: \cM \to \cN$. Then $D_{BS}(id\|E) = \log [\cM:\cN]$.
\end{proposition}

\begin{proof}
a) We let $d^2 = [\cM:\cN]$ and we first show $D_{BS}(id \|E)  \ge \log d^2$ using the variational definition \eqref{channeldiv} for the channel divergence 
in the case of the BS divergence. We let $e$ be the Jones projection for the inclusion, i.e. $\cM$ is generated by $\cN$ and $e$. Then $E(e) = d^{-2}1$. 
Let $\pi$ be the representation of $\cM \odot \cM^{op}$
coming from the standard bimodule $L^2(\cM)$ with underlying Hilbert space $\sH$. Recall that for $\xi \in \sH$ we have by definition
$\varphi_{\xi,E,\pi}(m \otimes m^{op}) = \langle \xi, E(m)J(m^{op})^*J\xi\rangle$ for the quantity appearing in \eqref{channeldiv} for the channel $E:\cM \to \cN$.
We use this bimodule in the variational characterization of proposition \ref{BS} involving a supremum over $n \in \mathbb N$ 
and admissible step functions $(1/n,\infty)\xrightarrow{x} \cM \odot \cM^{op}$ (as well as $y_t:=1-x_t$). We obtain a lower bound by constructing a specific step function 
$x_n$ for each $n \in \mathbb N$ and show that the limit $n \to \infty$ of the variational expression in proposition \ref{BS} tends to a quantity that is at least $\log d^2$. 

For this, we  choose a standard vector $\Omega \in \sH$ for $\cM$ and let $\xi:=e\Omega/\|e\Omega\|$. We also let 
\begin{equation}
x_{t} := 
\begin{cases}
1 - \tfrac{t}{t+d^{-2}}e \otimes 1 & \text{$1/n \le t \le n$,}\\
0 & \text{$t>n$}
\end{cases}
\end{equation}
and we let $y_{t} = 1-x_{t}$. Since $e\xi=\xi$ and $E(e) = d^{-2}1$, we get
\begin{equation}
\varphi_{\xi,E,\pi}(y_{t}^* y_{t}^{}) = \frac{t^2 d^{-2}}{(t+d^{-2})^2}, \quad
\varphi_{\xi,id,\pi}(x_{t}^* x_{t}^{}) = \frac{d^{-4}}{(t+d^{-2})^2}
\end{equation}
in the range $t \le n$. This gives us
\begin{equation}
\begin{split}
&
\int_{1/n}^\infty \{ \varphi_{\xi,id,\pi}(x_{t}^* x_{t}^{}) + \frac{1}{t}  \varphi_{\xi,E,\pi}(y_{t}^* y_{t}^{}) \} \frac{dt}{t}\\
&= \int_{1/n}^{n} \bigg( \frac{d^{-4}}{(t+d^{-2})^2}  + \frac{d^{-2} t}{(t+d^{-2})^2} \bigg) \frac{dt}{t} + \frac{1}{n}\\
&\le \frac{1}{n} + d^{-2} \int_{1/n}^\infty \frac{dt}{t(t+d^{-2})} \\
&=\log n + \log(n^{-1} + d^{-2}) + \frac{1}{n}.
\end{split}
\end{equation}
Then it follows from the variational characterization of the BS divergence (proposition \ref{BS}) that
\begin{equation}
\begin{split}
&D_{BS}(\varphi_{\xi,id,\pi}\|\varphi_{\xi,E,\pi})\\
&\ge \sup_n  \left( \log n -
\int_{1/n}^\infty \{ \varphi_{\xi,id,\pi}(x_{t}^* x_{t}^{}) + \frac{1}{t}  \varphi_{\xi,E,\pi}(y_{t}^* y_{t}^{}) \} \frac{dt}{t} 
\right) \\
&\ge \lim_n  \left( \log n -
\int_{1/n}^\infty \{ \varphi_{\xi,id,\pi}(x_{t}^* x_{t}^{}) + \frac{1}{t}  \varphi_{\xi,E,\pi}(y_{t}^* y_{t}^{}) \} \frac{dt}{t} 
\right) \\
&\ge \lim_n \bigg( \log n -(\log n + \log(n^{-1} + d^{-2}) + \frac{1}{n}) \bigg) = \log d^2.
\end{split}
\end{equation}
By the variational characterization of the channel divergence as a supremum of $D_{BS}(\varphi_{\xi,id,\pi}\|\varphi_{\xi,E,\pi})$ over triples $(\cA,\pi,\xi)$
we therefore have $D_{BS}(id\|E) \ge \log d^2$. 

b) 
The conditional expectation satisfies the Pimsner-Popa bound $E \ge d^{-2} id$ \cite{Hia88,PiPo86}. Let $\varepsilon >0$. Then we can choose a triple 
$(\pi,\cA,\xi)$  (consisting of a von Neumann algebra $\cA$, binormal representation $\pi$ on $\sH$ of $\cM \odot \cA^{op}$, and unit vector $\xi$ in $\sH$)  
an $n \in \mathbb N$,
and an admissible step function $(1/n,\infty)\xrightarrow{x} \cM \odot \cA^{op}$ such that 
the supremum in the variational definition \eqref{channeldiv} is saturated up to tolerance $\varepsilon$:
\begin{equation}
\begin{split}
& D_{BS}(id\|E)-\varepsilon \\
&\le \log n - \int_{1/n}^\infty 
\{ \varphi_{\xi,id,\pi}(x_t^{*}x_t^{})+\frac{1}{t}\varphi_{\xi,E,\pi}(y_t^*y_t^{})) \} \frac{dt}{t} \\
&\le \log n - \int_{1/n}^\infty 
\{ \varphi_{\xi,id,\pi}(x_t^{*}x_t^{})+\frac{1}{td^2}\varphi_{\xi,id,\pi}(y_t^*y_t^{}) \} \frac{dt}{t} \\
&= \log d^2 + \log (nd^{-2}) - \int_{1/(nd^{-2})}^\infty 
\{ \varphi_{\xi,id,\pi}(x_t^{*}x_t^{})+\frac{1}{t}\varphi_{\xi,id,\pi}(y_t^*y_t^{}) \} \frac{dt}{t} 
\end{split}
\end{equation}
The right side is $\le \log d^2 + D_{BS}(id\| id)  = \log d^2$ 
using the variational definition again. Since $\varepsilon > 0$ can be as small as we like we have shown $D_{BS}(id\| E) \le \log d^2$.
We have already shown $D_{BS}(id\|E) \ge \log d^2$ in a) so the proof  is complete.
\end{proof}

\section{Applications to QFT}

\subsection{Algebraic QFT}

\noindent
{\bf Axioms:} See \cite{Haa12} as a general reference. In the preceding sections we have described properties of the channel divergence $D_{BS}$ in the general context of 
von Neumann algebras. In the context of local QFT, one has additional structure due to spacetime localization, and it turns out that 
this structure plays very nicely with the notion of channel divergence.
We restrict to the setting of Minkowski spacetime $({\mathbb R}^n,\eta)$ for $n \ge 2$. 

A causal diamond $O$ is 
the causal completion of an open, simply connected subset $U$ with compact closure of a Cauchy surface, where the 
causal structure is induced by the Minkowski metric.
A QFT in the algebraic setting is an assignment of simply connected causal diamonds to von Neumann factors $O \mapsto \cA(O)$ represented on the 
same Hilbert space $\sH$, subject to the following conditions:

\begin{enumerate}
\item[a1)] (Isotony) $\cA(O_1) \subset \cA(O_2)$ if $O_1 \subset O_2$. We write $\cA = \overline{\bigcup_O \cA(O)}$ with completion in the operator norm.
\item[a2)] (Causality) $[\cA(O_1),\cA(O_2)]=\{0\}$ if $O_1$ is space-like related to $O_2$. 
\item[a3)] (Relativistic covariance) For each $g \in \widetilde{{\rm P}}$ 
covering\footnote{The covering group is needed to describe non-integer spin.} a Poincar\'e transformation 
$(\Lambda,a) \in {\rm P} = {\rm SO}_+(n-1,1) \ltimes \mathbb{R}^{n}$, there is an automorphism $\alpha_g$ on $\cA$ 
such that $\alpha_g \cA(O) = \cA(\Lambda O+a)$ for all causal diamonds $O$ and such that 
$\alpha_g \alpha_{g'} = \alpha_{gg'}$ and $\alpha_{(1,0)}=id$ is the identity.
\item[a4)] (Vacuum) There is a strongly continuous positive energy representation $g \mapsto U(g)$ on $\sH$ implementing 
$\alpha_g(a) = U(g) a U(g)^*$ for all $a \in \cA$. There is a vector $\Omega$ (the vacuum) which is cyclic for $\cA$ and 
such that $U(g)\Omega = \Omega$ for all $g \in \widetilde{{\rm P}}$. 
Positive energy means that if $x \in {\mathbb R}^n \subset {\rm P}$ is a translation by $x$, we can write
\begin{equation}
U(x) = \exp(-i \eta(P,x)),
\end{equation}
and the vector generator $P=(P^0,P^1,\dots,P^{n-1})$ has spectral values $p$ in the forward lightcone
 $p \in \bar{V}^+ = \{ p \in \mathbb{R}^n \mid \eta(p,p) \ge 0, p^0>0\}$.
 \item[a5)] (Additivity) Let $O_i$ be a family of causal diamonds such that $O = \cup_i O_i$. Then $(\cup_i \cA(O_i))'' = \cA(O)$.
\end{enumerate}

For technical purposes, we also impose a ``nuclearity condition.'' The main purpose of that condition is to ensure a certain regularity on the theory, and several closely related 
versions of such a condition have been proposed. In so far as we can see, many of these 
would more or less all be equally good for our purposes. For definiteness, we impose~\cite{buchholz1986causal}:

\begin{enumerate}
\item[a6)] (BW-nuclearity) Let $A$ be a ball of radius $r$ in Cauchy surface, and let $O_r$ be the corresponding causal diamond. Consider the map
\begin{equation}
\Theta_{\beta,r}: \cA(O_r) \to \sH \ , \quad a \mapsto e^{-\beta H}a\Omega \ ,
\end{equation}
where $\beta>0$ and where $H=P^0$ is the Hamiltonian, i.e. the time-component of $P$ in item a4). 
It is required that there exist positive constants $s>0$ and $c = c(r)>0$ such that for $r>0, \beta>0$ we have
$
\| \Theta_{\beta, r} \|_1 \le e^{(c/\beta)^s} \ .
$
Here we use the nuclear 1-norm discussed further e.g. in~\cite{pietsch2022nuclear}. 
\end{enumerate}

We now comment on two well-known important consequences of these results for our analysis, see \cite{Haa12} for further details and references. 
First, by the Reeh-Schlieder theorem, $\Omega$ is cyclic and separating for each $\cA(O)$, so the vacuum automatically provides a standard 
form for each local von Neumann algebra. Secondly, each $\cA(O)$ is a hyperfinite factor of type III$_1$ \cite{buchholz1987universal} 
which is a unique object up to von Neumann isomorphism by 
\cite{haagerup1987conne}. As a consequence, we can apply all of our results on the channel divergences $D_{BS}$ to the local algebras $\cA(O)$.

It is important to stress that a priori, $\cA(O)$ is defined only for causal diamonds associated with simply connected subsets of a Cauchy surface. If $K$ is any open, causally 
complete subset of $\mathbb{R}^n$, we could define either
\begin{equation}
\label{ABdef}
\cA(K) = (\vee_{O \subset K} \cA(O))'', \quad \text{or} \quad 
\cB(K) = (\vee_{O' \subset K'} \cA(O'))'.
\end{equation}
In either case, a prime on a region $O$ or $K$ means the causal complement. For topologically trivial causal diamonds $O$ with compact closure 
it is a result that $\cA(O') = \cA(O')$ (Haag duality), so by a5), $\cA(O) = \cB(O)$ for topologically trivial causal diamonds. 
Either $\cA$ or $\cB$ gives a net in the above sense with the possible exception of condition a5) in the case of $\cB$. 
$\cB(K)$ is in general strictly bigger for topologically non-trivial regions $K$ than $\cA(K)$.

\medskip
\noindent
{\bf DHR-Representations:} See \cite{buchholz1982locality, doplicher1971local, doplicher1969fields, Haa12}. The Hilbert space $\sH$ may be considered as the defining (vacuum) representation of the net, but it is physically relevant to also consider other representations. 
We shall consider representation $\pi$ of $\cA$ on a Hilbert space $\sH_\pi$ which are ultraweakly continuous when restricted to 
any $\cA(O)$ and which satisfy:
\begin{itemize}
\item (DHR-selection criterion) \cite{doplicher1971local,doplicher1969fields}
$\pi |_{\cA(O)' \cap \cA}$ is unitarily equivalent to the vacuum representation for some $O$.
\item (BF-selection criterion) \cite{buchholz1982locality}
The automorphisms $\alpha_g$ in a3) are unitarily implemented in $\pi$, i.e. there exists a strongly continuous positive energy representation 
$U_\pi(g)$ such that $\pi(\alpha_g(a)) = U_\pi(g) \pi(a) U_\pi(g)^*$ such that the generator $P_\pi$ of translations $U_\pi(x) = \exp(-i \eta(P_\pi,x))$ on $\sH_\pi$
has an isolated mass shell in its spectrum, i.e. $spec(P_\pi) \subset \{ p: \eta(p,p) = M^2, p^0>0\} \cup  \{ p: \eta(p,p) \ge m^2, p^0>0\}$ for some $m^2 > M^2 >0$. 
\end{itemize}

If we let $V$ be a unitary implementing the unitary equivalence in the first item, then $\rho(a):= V^* \pi(a) V$ is an endomorphism of $\cA$ such that 
\begin{equation}
\rho|_{\cA(O)' \cap \cA} = id. 
\end{equation}
One says that $\rho$ is a {\em localized endomorphism (in $O$)} for this reason. Furthermore, $\rho$ is {\em transportable} in the following 
sense. Let $O_1:=O$, $\rho_1:= \rho$ and let $O_2$ be another causal diamond. Then there exists a unitary $u_{21} \in \cA(O_1) \vee \cA(O_2)$ such that 
$Ad u_{21} \circ \rho_1 =: \rho_2$ is an endomorphism satisfying the DHR- and BF- selection criteria that is localized in $O_2$. We will refer to the endomorphisms
arising from the selection criteria above as a localized, transportable endomorphism. 

Let $\rho$ be a transportable irreducible endomorphism localized in some $O$. As is known, the selection criteria imply a considerable amount 
of further algebraic structure associated with $\rho$. First, we have a so-called conjugate 
transportable endomorphism $\bar \rho$  together with solutions $r,\bar r \in \cA(O)$ and $d_\rho \ge 1$ to the intertwining
\begin{equation}
\rho\bar \rho(a) \bar r = \bar ra, 
\quad 
\bar \rho \rho(a)r = ra \quad (a \in \cA)
\end{equation} 
and conjugacy relations
\begin{equation}
\label{conjugacy}
r^* r = d_\rho1, \quad \bar r^* \bar r = d_\rho1, \quad r^* \bar \rho(\bar r) = 1 = \bar r^* \rho(r).
\end{equation}
A left inverse of $\rho$ is given by $\Psi_\rho(a):=d^{-1}_\rho r^* \bar \rho(a) r$. The Jones projection for the extension $\cA(O)$ of $\rho(\cA(O))$ 
is given by $e_\rho=d^{-1}_\rho \bar r \bar r^*$ and the minimal conditional expectation is $E_\rho: \cA(O) \to \rho(\cA(O))$ is given by $E_\rho = \rho \circ \Psi_\rho$.
$d_\rho$ is referred to as the ``statistical dimension'' of $\rho$. By the index-statistics theorem \cite{Lon89}, $d_\rho=[\cA(O):\rho(\cA(O))]^{1/2}<\infty$.
Similar constructions apply to reducible endomorphisms/representations. 

For a variant of this theory for conformal field theories in $n=2$ spacetime dimensions 
see \cite{fredenhagen1989superselection,Lon89}. 

\subsection{Complexity of channels in AQFT}

Let $T$ be a completely positive map of the quasi-local algebra $\cA$,\footnote{Note that $\cA$ is a $C^*$- and not a von Neumann algebra, but the notion of 
completely positive map is still defined.} such that, for some sufficiently large causal diamond $O$, 
it restricts to a channel of $\cA(O)$. By \cite[Theorem 2.10]{Lon18}, we may write 
\begin{equation}
\label{Ta}
T(a) = v^*\theta(a)v, \quad a \in \cA(O)
\end{equation}
where $v$ is an isometry of $\cA(O)$ and $\theta$ is an endomorphism of $\cA(O)$. This motivates the following definition. 

\begin{definition}
A channel $T: \cA \to \cA$ is called localized and transportable if it is of the form \eqref{Ta} for some localized (in some causal diamond $O$) 
transportable endomorphism $\theta$ and some isometry $v \in \cA(O)$.
\end{definition}

\begin{remark}
1) Note that by definition, $T|_{\cA(O)' \cap \cA} = id$, i.e. $T$ is the identity in the causal complement of $O$. 

2) It is easy to see that the set of localized and 
transportable channels is stable under composition, i.e. the composition is again of the form \eqref{Ta}. 
It is also closed under convex combinations: Let $T_i$ be localized, transportable channels of the form \eqref{Ta}
with $v_i, \theta_i$, and $p_i$ a 
probability distribution on a finite set. Since $\cA(O)$ is type III \cite{Buc87}, 
there are isometries $a_i$ in $\cA(O)$ satisfying the Cuntz algebra relations \eqref{Cuntz}. 
Then set $v = \sum \sqrt{p_i}  a_iv_i$ and $\theta(m) = \sum a_i \theta_i(m) a_i^*, m \in \cA(O)$. It 
follows that $\theta$ is a localized, transportable endomorphism of $\cA$, that $v$ is an isometry of $\cA(O)$, and that $\sum p_iT_i$ is of the form  \eqref{Ta}.

3) One may generalize the definition to channels between two nets $\cA, \cB$.
\end{remark}

Let $x \in \mathbb{R}^n$, let $O+x$ be the translate of $O$
and let $\alpha_x(a) = U(x)^* a U(x)$ be the translate of an element  $a \in \cA(O)$ to $\cA(O+x)$ as in a3). We consider $T_x = \alpha_x \circ T \circ \alpha_{-x}$
as a channel of $\cA(O+x)$. Then 
\begin{equation}
T_x(a) = \alpha_x(v)^* U(x)\theta(\alpha_{-x}(a))U(x)^* \alpha_x(v)
\end{equation}
Since $\theta$ is by assumption an endomorphism satisfying the DHR- and BF selection criteria, translations are implemented in the sector $\theta$ by 
a strongly continuous group of unitaries $U_\theta(x), x \in \mathbb{R}^n$ so we have $\theta(\alpha_{-x}(a)) = U_\theta(x)^* \theta(a) U_\theta(x)$. 
Furthermore $u(x) = U(x) U_\theta(x)^*$ is an element of $\cA(O) \vee \cA(O+x)$ transporting $\theta$ to an endomorphism $\theta_x = Ad_{u(x)} \circ \theta$ localized in 
in $O+x$, and we have, with $v_x = \alpha_x(v) \in \cA(O+x)$, 
\begin{equation}
\label{Tax}
T_x(a) = v_x^* \theta_x^{} (a) v_x^{}, a \in \cA(O+x),
\end{equation}
i.e. it has the same form as \eqref{Ta} but with $\theta_x, v_x$ now localized in $O+x$. 

We now make a proposal for the complexity of a channel in algebraic quantum field theory. 

\begin{definition}
The complexity of a localizable and transportable channel $T$ is defined as 
\begin{equation}
\label{eqO}
c(T) = D_{BS}(id|_{\cA(O)}\|T|_{\cA(O)}),
\end{equation}
where $O$ is any sufficiently large causal diamond such that $T|_{\cA(O)' \cap \cA} = id$.
\end{definition}

To be precise, we should demonstrate:

\begin{lemma} 
The definition of $c(T)$ does not depend on the 
 sufficiently large causal diamond $O$ chosen in \eqref{eqO}. 
\end{lemma}
\begin{proof}
Let $O_1,O_2$ be causal diamonds such that 
$T|_{\cA(O_i)' \cap \cA} = id$ and let $O$ be the causal completion of $O_1 \cup O_2$.
 Let $\cA(O_1)^c := \cA(O_1)' \cap \cA(O)$, let $\cM_n^c$ be a net 
of finite dimensional type I algebras exhausting $\cA(O_1)^c$, and let $\cM_n$ be a net of finite-dimensional type I algebras 
exhausting $\cA(O_1)$, which exist as a consequence of requirement a6), see 
\cite{Buc87}. Then $(\cup_n \cM_n^c \vee \cM_n^{})''=\cA(O)$ and by the martingale property for $D_{BS}$
we have 
\begin{equation}
\begin{split}
D_{BS}(id|_{\cA(O)} \| T|_{\cA(O)}) 
=& \lim_n D_{BS}(id |_{\cM_n^c \vee \cM_n^{}} \| T|_{\cM_n^c \vee \cM_n^{}})\\
=& \lim_n D_{BS}(id |_{\cM_n^{}} \otimes id |_{\cM^c_n} \| T|_{\cM_n^{}} \otimes id |_{\cM^c_n})\\
=& \lim_n D_{BS}(id |_{\cM_n^{}} \| T|_{\cM_n^{}})\\
=&D_{BS}(id |_{\cA(O_1)} \| T|_{\cA(O_1)}).
\end{split}
\end{equation} 
In the second equality, we used that $\cM_n^{} \vee \cM_n^c \cong \cM_n^{} \otimes \cM_n^c$ as 
von Neumann algebras because $\cM_n$ and $\cM^c$ are finite-dimensional and that 
$T$ acts trivially on $\cM_n^c$ by locality. In the third step we used external additivity of $D_{BS}$. In the 
last step we used again the martingale property. The same could be shown for $O_1 \to O_2$. 
Thus the definition of $c(T)$ is independent of whether we take $O_1$ or $O_2$
in \eqref{eqO}.
\end{proof}

\begin{theorem}
\label{mainthm}
The complexity $c$ has the following properties ($T,T_i$ localized, transportable channels): 
\begin{enumerate}
\item  (Identity) $c(id)=0$.

\item (Internal subadditivity) $c(T_1 \circ T_2) \le c(T_1) + c(T_2)$.

\item (Convexity) Let $\{p_i\}$ be a probability distribution on a finite set. Then 
$c(\sum p_i T_i) \le \sum p_i c(T_i)$.

\item (Locality) Let $T_1$ and $T_2$ be channels localized in spacelike related causal diamonds with strictly positive distance. 
Then $c(T_1 \circ T_2) = c(T_1) + c(T_2)$.

\item ($N$-ary local measurement) Let $M(a) = \sum_i e_i a e_i$ be the channel describing an $N$-ary 
local measurement associated with the $N$ mutually orthogonal non-trivial projections $e_i \in \cA(O)$, $\sum e_i=1$. Then
$M$ is localized and transportable and $c(M) = \log N$. 

\item (Net extensions) Let $\cB$ be a net extending $\cA$ \cite{longo1995nets}
with corresponding conditional expectation $E$. Then 
\begin{equation}
c(E) = \log [\cB(O):\cA(O)].
\end{equation}

\item (Localized transportable endomorphisms I) 
Let $\rho$ be a transportable localized transportable endomorphism with conditional expectation $E_\rho$
and statistical dimension $d_\rho$. Then $E_\rho$ is a localized transportable channel and 
\begin{equation}
c(E_\rho) = \log d_\rho^2.  
\end{equation}

\item (Translations) If $T_x$ is the translate of $T$ by $x \in \mathbb{R}^n$ as in \eqref{Tax}, then $c(T_x) = c(T)$.

\item (Localized transportale endomorphisms II) If $\rho \neq id$ is a localized  transportable endomorphism of $\cA(O)$, then $c(\rho) =\infty$.

\item (Local unitaries) Let $u \in \cA(O)$ be a unitary and $U(a) = u^*au$ be the corresponding channel on $\cA(O)$. Then if $U \neq id$, we have
$c(U)=\infty$.
\end{enumerate}
\end{theorem}
\begin{proof}
1)-3),6),10) are taken from section \ref{sec:chdiv}. 

4) Let $T_i$ be localized in $O_i$, where $O_1$ and $O_2$ are spacelike related with strictly positive distance. 
By locality $T_1 \circ T_2(a_1a_2) = T_1(a_1)T_2(a_2), a_i \in  \cA(O_i)$.
Then the BW-nucelarity assumption a6) implies the split property for the algebras $\cA(O_1)$ and $\cA(O_2)$, see \cite{Buc87} or 
\cite[Chapter V.5.2]{Haa12} as a general reference. So there is a unitary $W:\sH \to \sH \otimes \sH$ such that $W^*(a_1 \otimes a_2)W = a_1 a_2$
and consequently $T_1 \circ T_2|_{\cA(O_1)\vee \cA(O_2)} = Ad_W \circ (T_1 \otimes T_2) \circ Ad_{W^*}$, where $T_1 \otimes T_2$ is 
the tensor product channel on $\cA(O_1) \otimes \cA(O_2)$ and $Ad_W X = W^* X W$. In particular, the map $T_1 \circ T_2$ is 
normal on $\cA(O_1) \vee \cA(O_2)$. Then, by applying internal subadditivity twice
\begin{equation}
D_{BS}(id\| T_1 \circ T_2) \ge D_{BS}(Ad_W\| Ad_W \circ (T_1 \otimes T_2)) \ge D_{BS}(id\| T_1 \otimes T_2)
\end{equation}
Since $W$ is unitary, we have the reverse inequality by the same argument backwards, so $c(T_1 \circ T_2) = c(T_1 \otimes T_2) = c(T_1) + c(T_2)$, 
by external additivity. 

5) This follows from section \ref{sec:chdiv}; we only need to show that $H$ is localized and transportable. Since $\cA(O)$ is properly infinite, 
there are isometries $a_i, i=1, \dots, N$ in $\cA(O)$ satisfying the Cuntz algebra relations \eqref{Cuntz}. 
Then we set $v^* = \sum_i e_i a_i^*$ and $\theta(m):=\sum_j a_j m a_j^*, m \in \cA(O)$. It follows that $v^* v = 1$, that $\theta$ is 
a localized, transportable endomorphism, and that $M(m) = v^*\theta(m)v$, as desired. 

7) The formulas $E_\rho(a) = d_\rho^{-1} \rho(r)^* \rho\bar\rho(a) \rho(r)$ and $r^*r=d_\rho 1$ show that $E_\rho$ is localized and transportable
(with $\theta = \rho\bar \rho$, $v=d^{-1/2}_\rho \rho(r)$).
The formula follows from proposition \ref{prop:E} because $d_\rho = [\cA(O):\rho\cA(O)]^{1/2}$ by the index-statistics theorem
\cite{Lon89}.

8) We only need to show $c(T_x) = c(T)$. By applying internal subadditivity twice we see
$D_{BS}(id\| \alpha_x \circ T \circ \alpha_{-x}) \ge D_{BS}(\alpha_{-x}\|T \circ \alpha_{-x}) \ge D_{BS}(id\|T)$.
We can also get the reverse inequality by running this argument backwards, thereby proving the claim. 

9) We view $E_\rho, \Psi_\rho, \rho$ as maps on some $\cA(O)$ such that $\rho$ is localized within $O$.
By using subadditivity and the formula $\Psi_\rho \circ \rho = id$ twice:
\begin{equation}
D_{BS}(id\|\rho) \ge D_{BS}(\Psi_\rho \|\Psi_\rho \circ \rho) = D_{BS}(\Psi_\rho\|id) \ge D_{BS}(\Psi_\rho \circ \rho\|\rho) = D_{BS}(id\|\rho).
\end{equation}
So we must have equality in each step. 

a) We first assume $d_\rho>1$. To get a lower bound (actually $+\infty$) on $D_{BS}(\Psi_\rho \|id)$, we proceed as in the proof of part a) of proposition \ref{prop:E}, noting that the Jones projection is 
$e_\rho = d_\rho^{-1} \bar r \bar r^*$, and $\Psi_\rho(a) = d^{-1}_\rho r^* \bar \rho(a) r$, with $r,\bar r, \bar \rho$ as in the conjugacy relations \eqref{conjugacy}. 
Then, as is well-known, $\Psi_\rho(e_\rho) = d_{\rho}^{-2} 1$, again by the conjugacy relations. As our trial function, we now choose
\begin{equation}
x_{t} := 
\begin{cases}
\tfrac{t^{-1}}{t^{-1}+d^{-2}_\rho}e_\rho \otimes 1 & \text{$1/n \le t \le n$,}\\
0 & \text{$t>n$}
\end{cases}
\end{equation}
and we let $y_{t} = 1-x_{t}$. The rest is similar as in as in the proof of part a) of proposition \ref{prop:E}: We see that with this trial function
\begin{equation}
\varphi_{\xi,\Psi,\pi}(x_t^{} x_t^*) + \frac{1}{t} \varphi_{\xi,id,\pi}(y_t^{} y_t^*) = 
\begin{cases}
\frac{1}{d^2_\rho + t}, & \text{for $1/n \le t \le n$,}\\
\frac{1}{t^2} & \text{$t>n$}
\end{cases}
\end{equation}
and thereby:
\begin{equation}
\begin{split}
D_{BS}(id\|\rho)=&D_{BS}(\Psi_\rho\|id) \\
=& \sup_n \sup_{x,\xi} \left( \log n - \int_{1/n}^\infty \{ 
\varphi_{\xi,\Psi,\pi}(x_t^{} x_t^*) + \frac{1}{t} \varphi_{\xi,id,\pi}(y_t^{} y_t^*)
\} \frac{dt}{t}
\right) \\
\ge& \sup_n  \left( \log n - \int_{1/n}^n \frac{1}{d^2_\rho + t} \frac{dt}{t} - \int_n^\infty \frac{dt}{t^2}
\right)\\
=& \sup_n \frac{1}{d_\rho^2} \left( (d_\rho^2-1) \log n - \log d^2_\rho + \log(1 + \frac{d_\rho^2}{n}) - \log(1+\frac{1}{d_\rho^2 n}) - \frac{d_\rho^2}{n} \right) \\
=& \infty
\end{split}
\end{equation} 
since $d_\rho>1$. 

b) If $d_\rho=1$, then $\rho(\cA(O))=\cA(O)$ and $\rho$ is an automorphism. Viewed as an automorphism of $\cA$, we have 
$\rho(b)=b$ for any $b \in \cA(K)$ so long as the causal diamond $K$ is contained in $O'$. Define the state $\psi = \omega \circ \rho^{-1}$
on $\cA(O)$, where $\omega$ is the vacuum state, and let $\Psi$ be the representer of $\psi$ in the natural cone of $\Omega$. 
Then $Ua\Omega := \rho(a)\Psi$ defines a unitary $U$ and if $b \in \cA(K)$ then clearly $bUa\Omega = b\rho(a)\Psi = 
\rho(ba) \Psi = Uba\Omega$. It follows that $U \in (\cup_{K \subset O'} \cA(K))' = \cA(O')' = \cA(O)$ by Haag duality. 
Thus $\rho$ is inner when restricted to $\cA(O)$ and hence $c(\rho)=\infty$ unless $\rho = id$ by item 10). 
\end{proof}

\begin{remark}
The specific local structure of QFT enters in an indirect way in theorem \ref{mainthm} because it entails that the 
local algebras $\cA(O)$ are hyperfinite, satisfy the split property, are properly infinite (in fact type III$_1$), and have a cyclic and separating vector -- 
the vacuum by the Reeh-Schlieder theorem. This is used in various combinations in the proofs 
of these properties. Nevertheless the 
given specific forms of the axioms are probably not totally essential; in particular it seems unlikely 
that the specific properties of Minkowski spacetime 
are crucial. What is missing in theorem \ref{mainthm} is a property linking the complexity of a channel to notions of energy transfer/cost. 
\end{remark}

\begin{example}
The simplest setting for an inclusion of nets as in 6) in theorem \ref{mainthm} is when $\cA = \cB^G$ is the fixed point net under some finite ``internal gauge'' group $G$, say $\mathbb{Z}_N$ for definiteness. 
$\mathbb{Z}_N$ is acting by unitaries $U(g), g \in \mathbb \mathbb{Z}_N$ on a common Hilbert space $\sH$ for both nets and each $U(g)$ commutes with the 
unitaries implementing the spacetime symmetries. The conditional expectation is just the group average $E(b) = N^{-1} \sum_g U(g) b U(g)^*$. Let $\chi_k$ be 
a  character on $\mathbb{Z}_N$. Then $P_k = \sum_g \chi_k(g) U(g)$ is a projection and $\sH_k = P_k \sH$ is an invariant subspace for each $\cA(O)$ which can be 
seen as an irreducible representation $\pi_k$ of $\sH$. 
\end{example}

\begin{example}
A more complicated class of examples for 6) in theorem \ref{mainthm} arises in $n=2$ conformal CFTs. Start with a conformal net $\cV_c$ on circle $S^1$, e.g. the Virasoro net for some central charge $c<1$. 
The conformal net $\cA$ is obtained as $\cV_c \otimes \cV_c^{op}$ identifying causal diamonds on $\mathbb{R}^2$ with Cartesian products of intervals. Then one can 
obtain an extension for $\cA$ from the set of highest weight representations $\cV_c$ (which give transportable irreducible endomorphisms $\mu, \nu, \dots$ 
on the Virasoro net on the real line) by starting from the representation 
$\sum_{\mu,\nu} Z_{\mu,\nu} \mu \otimes \nu^{op}$ of $\cV_c \otimes \cV_c^{op}$, with $Z_{\mu,\nu}$ the multiplicities in the torus partition functions. 
If $\cB$ is the corresponding net extending $\cA$, the index is given by $[\cB : \cA] =\sum_\mu d_\mu^2$. For details, see \cite{rehren2000canonical}.
\end{example}

\begin{example}
A situation similar to to the previous example arises in local gauge theories of Yang-Mills type in $n=4$ dimensions based on a compact local gauge group $G$: 
Take $K$ to be the causal completion of a solid torus in a Cauchy-surface. Then we have $\cA(K), \cB(K)$ as in \eqref{ABdef}. As argued in \cite{casini2020entanglement}, 
we should have $[\cB(K):\cA(K)] = {\rm dim} Z(G)$ where $Z(G)$ is the center of the gauge group (e.g. $\mathbb{Z}_N$ in the case of $G=SU(N)$). Thus
we relate a property of the gauge group to the complexity of the conditional expectation $E:\cB(K) \to \cA(K)$. An intuitive reasoning for what $E$ does in 
terms of `t Hooft and Wilson loops is given in \cite{casini2020entanglement}.
\end{example}

\begin{example}
\label{DHRexp}
An example for localized endomorphisms $\rho$ as in 7) in theorem \ref{mainthm} in a free field theory is the following \cite[Section 4.7]{hollands2018entanglement}.
Consider a real $N$-component free complex Klein-Gordon quantum field $\phi_I(x), I=1, \dots, N$ in $n=4$ dimensions. We get a net $\cA$ 
of all observables that are invariant under the obvious action of the $SU(N)$-symmetry. 
Consider a tensor $T^{I_1 \dots I_k}, I_j=1, \dots, N$ whose symmetry properties under index permutations are
characterized by some Young-tableau  $\mathbf{\lambda}=(\lambda_1,\dots,\lambda_s)$ where $\lambda_1 \ge \dots \ge \lambda_s$
and where $\lambda_i$ is the number of boxes in the $i$-th row.
Next, take testfunctions $f_I$ with support in a causal diamond $O$. Define
\begin{equation}
\Psi = C \sum_{I_1, \dots, I_k=1}^N T^{I_1 \dots I_k} \phi_{I_1}(f_1) \dots \phi_{I_k}(f_k)\Omega  \ ,
\end{equation}
where $\phi_I(f)=\int \phi_I(x) f(x) d^4 x$ are the smeared KG quantum fields and $\Omega$ the vacuum vector and $C$ is a factor 
such that $\|\Psi\|=1$. 
Let $\dim(\mathbf{\lambda})$ be the dimension of the space of tensors with Young-tableau symmetry $\mathbf{\lambda}$.
By DHR theory \cite{doplicher1971local,doplicher1969fields}, there exist a localized (within $O$), transportable endomorphism $\rho$ of the net
such that 
\begin{equation}
\langle \Omega, \rho(a) \Omega \rangle = \langle \Psi, a \Psi \rangle\quad \text{for $a \in \cA$}
\end{equation}
and the 
statistical dimension $d_\rho$ of this $\rho$ equals the Young-tableau dimension $\dim(\mathbf{\lambda})$.
It is given by a standard formula in terms of the shape of the Young tableau, see e.g. \cite{fischler1981young}, so we obtain from 7) in theorem \ref{mainthm}
in this example,
\begin{equation}
c(E_\rho)= \log \dim (\mathbf{\lambda})^2 = 2 \log \frac{\prod_i (\lambda_i+N-i) \prod_{i<j} (\lambda_i-\lambda_j-i+j)}{(N-1)!} \ ,
\end{equation}
In the following example diagram $\mathbf{\lambda}:$ 
{\tiny
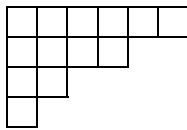
\begin{ytableau}
~&&&&&\\
~&&&&\none&\none\\
~&&\none&\none&\none&\none\\
~&\none&\none&\none&\none&\none
 \end{ytableau}
 }
with $k=13$ and $N=10$, the right side is $2 \log 135$.
\end{example}

\vspace{1cm}

{\bf Acknowledgements:} SH is grateful to Nima Lashkari and Mark Wilde for discussions and  
to the Max-Planck Society for supporting the collaboration between MPI-MiS and Leipzig U., grant Proj.~Bez.\ M.FE.A.MATN0003. 
AR thanks Hideki Kosaki for discussions. Part of this work was carried out while he was visiting ITP, Leipzig U. 
He is grateful financial support and hospitality. AR was partially supported by Department of Excellence MatMod$@$TOV by MIUR.

{\bf Declarations:}
Data sharing not applicable to this article as no datasets were generated or analyzed during the current study. 
The authors are not aware of a conflict of interest on their side related to this article. 

\bibliography{Complexity_AQFT_160223}
\bibliographystyle{ieeetr}

\end{document}